\documentclass[journal,english]{IEEEtran}
\usepackage[utf8]{inputenc}
\usepackage[T1]{fontenc}
\usepackage{hyperref}
\usepackage{graphicx} 
\usepackage[style=ieee, doi=false, isbn=false, url=false]{biblatex}
\usepackage{commands}
\usepackage{dsfont}
\usepackage{orcidlink}
\usepackage{enumitem}
\usepackage{babel,csquotes}

\usepackage{tikz}
\usetikzlibrary{positioning, calc}

\title{Local Geometry of Least Squares for Unmixing Signals with Parameter-Dependent Dictionaries}

\author{Santos~Michelena\,\orcidlink{0009-0003-5082-5932},
        Maxime~Ferreira~Da~Costa\,\orcidlink{0000-0003-0073-8825},
        and~Jos\'e~Picheral\,\orcidlink{0000-0001-9571-8151}
\thanks{S. Michelena, M. Ferreira Da Costa, and J. Picheral are with the Laboratory of Signals and Systems, CentraleSup\'elec, Universit\'e Paris--Saclay, CNRS, 91190 Gif-sur-Yvette, France (e-mails: \{santos.michelena, maxime.ferreira, jose.picheral\}@centralesupelec.fr).}
\thanks{S. Michelena is also with iUMTEK, 91400 Orsay, France.}
\thanks{M. Ferreira Da Costa acknowledges support from ANR grant ANR-24-CE48-3094.}
}

\begin{document}

\maketitle

\begin{abstract}
Modeling signals as linear combinations of atoms from a dictionary is ubiquitous in modern signal processing. In the finite-dimensional setting, whenever atoms depend nonlinearly upon unknown parameters, the signal model is said to be \emph{separable}. In this work, we study least-squares reconstruction of separable signals and establish a unified theoretical framework for their analysis. We introduce the \emph{unmixing metric}, a distance that captures the distinct roles and sensitivities of linear and nonlinear parameters, and establish local convergence and stability guarantees under its topology. We then analyze variable projection from a geometric perspective, showing that it corresponds to restricting the optimization to the manifold of optimal linear parameters. This viewpoint provides a principled explanation for the improved algorithmic behavior of variable projection observed in practice, and produces sharp theoretical guarantees. The generic theory for separable problems is specialized to the case of point spread function (PSF) unmixing. We introduce a parametric notion of coherence and show that support separation directly controls both the size of the convergence region and the stability of recovery. Numerical experiments corroborate the theoretical predictions and demonstrate the practical relevance of the proposed framework.
\end{abstract}

\section{Introduction}

\IEEEPARstart{S}{ignal} models composed of linear combinations of nonlinear parameter-dependent atoms are said to be \emph{separable}. In separable models, nonlinear parameters $\bx$ determine a structured dictionary, or forward map $\A(\bx)$, and linear coefficients $\by$ determine mixture weights. Such signal models are ubiquitous in a myriad of fields of engineering, applied science, and physics, such as blind image deblurring~\cite{salzer2026varpro}, time-resolved spectroscopy~\cite{Mullen2007}, medical signal analysis~\cite{kovacs-var-pro} and dynamical systems~\cite{wang2024varproPDE}.

In this article, we consider recovery of nonlinear parameters and mixture weights via the classical least-squares formulation~\cite[Chapter 10]{Nocedal2006}. Given the nonlinear dependence of the dictionary on $\bx$, this estimation task falls into the framework of nonconvex optimization. In the nonconvex setting, spurious minimizers may arise; consequently, optimization algorithms are sensitive to initialization, and recovery guarantees are typically local. A substantial body of work, therefore, focuses on establishing convergence radii for specific optimization schemes in the least-squares setting. Classically, a theorem of Kantorovich \cite{KantorovichAkilov1982} establishes the convergence radius of Newton’s method for solving nonlinear equations. Extensions of this framework yield convergence guarantees for the Gauss–Newton method \cite{gauss-newton-kantorovich}, while more recent developments derive optimal convergence radii under generalized Lipschitz assumptions on the Jacobian of the residual \cite{ferreira_local_2010, ferreira_kantorovichs_2009,SILVA2016178}. Albeit optimal, the convergence radii established by these results are relevant only to the chosen optimization method.

In lieu of tying guarantees to a specific algorithm, convergence can be established in terms of strong convexity of the loss around the ground truth. Specifically, positivity and uniform boundedness of the Hessian together define the \emph{strong basin of attraction}, a region where the loss exhibits optimal curvature and smoothness. This geometric perspective offers insight into how the forward map influences problem conditioning, thus exposing the mechanisms driving recovery performance and yielding algorithm-free convergence radii. Moreover, once local convexity is established, stability of recovery follows directly.

Exploiting the separable structure, the method of variable projection eliminates the linear variables in closed form, reducing the problem to the nonlinear parameter space~\cite{variableProjection}. A substantial body of work reports improved numerical performance relative to joint and alternating optimization, including enhanced stability and faster local convergence in many practical settings \cite{varpro-meta-review, projected-hessian, oleary-varpro, kaufman1975variable}. These improvements are typically attributed to better conditioning of the reduced problem through the pseudo-inverse. Despite these algorithmic advantages, explicit study of the local geometry of projected objectives remains relatively scarce in the literature.

Separable models are closely related to several classical inverse problems in signal processing, including blind deconvolution~\cite{Chretien2020, moulines_maximum_1997, li2016identifiability, li2019multichannel}, bilinear inverse problems~\cite{Beinert_2019}, and nonlinear compressive sensing \cite{nonlin-compress}. Indeed, in these models, observations typically depend nonlinearly on a low-dimensional parameter vector and linearly on mixing coefficients, often leading to nonconvex optimization landscapes with structured geometry. Recent approaches achieve stable and scalable recovery under strong structural assumptions~\cite{traonmilin2020basins, traonmilin2024strong, costa_local_2023}. Related problems also admit elegant functional formulations~\cite{chamon}; however, the fact that each dictionary atom depends on a single underlying function renders such models less general than the generic separable framework considered herein.

We take particular interest in \emph{PSF unmixing}, where each dictionary atom corresponds to a kernel or point spread function (PSF) centered at a known location but parametrized by an unknown shape parameter. This setting encompasses a wide range of applications, such as super-resolution imaging~\cite{huang2008three}, sorting of neural recordings~\cite{knudson2014inferring}, and calibration-free laser-induced breakdown spectroscopy (CF-LIBS)~\cite{poggialini_catching_2023}, where deconvolution of spectral lines that have been broadened by diverse, challenging to model physical effects is a crucial step in the analysis.

In this context, we introduced in our previous work~\cite{michelena_conv_2025} a parametrized notion of coherence to bound the size of the strong basin of attraction in terms of the separation between PSF centers. In~\cite{michelena-icassp}, we applied this analysis to the variable projection case. However, in this work, the analysis hinged on expanding the Lipschitz constants of the coherence, thereby obfuscating the applicability and interpretability of the analysis. 
The article contributions are as follows:
\begin{itemize}[wide]
    \item We establish local convergence and stability guarantees for least-squares reconstruction of separable models, expressed in terms of the local geometry of the objective. Optimization is performed jointly over nonlinear and linear parameters. To capture their distinct roles and sensitivities, we introduce the \emph{unmixing metric} on the concatenated parameter $\btheta \coloneq (\bx, \by)$, and derive convergence radii and stability bounds with respect to this metric. The basin radius is characterized by spectral constants associated with the dictionary and its Fréchet derivatives.
    \item Furthermore, we analyze the method of variable projection from a geometric perspective. By means of the pseudo-inverse, the method implicitly restricts the linear parameter to an optimal value surface, making the improved conditioning intrinsic to the formulation. The associated basin radius follows from the observation that the variable projection Hessian coincides with the full Hessian restricted to the tangent space of this surface. This restriction induces a conditioning amplification factor, yielding a convergence radius equal to the ambient radius scaled by this factor and adjusted by the metric change due to reduced degrees of freedom.
    \item Finally, specializing to PSF unmixing, we show that the spectral constants governing the basin size reduce to a parametric notion of coherence, itself controlled by support separation. Consequently, the local convergence and stability of the high-dimensional nonlinear problem are determined by a single physically meaningful parameter.
\end{itemize}

The rest of the paper is organized as follows. Section~\ref{sec:local-geometry} introduces the separable model and the notion of a strong basin of attraction, and derives convergence and stability guarantees for least-squares and variable projection methods. Section~\ref{sec:psf-unmixing} specializes the analysis to PSF unmixing, expressing guarantees in terms of support separation via parametric coherence. Section~\ref{sec:numerics} presents numerical experiments validating the theory, and Section~\ref{sec:conclusion} concludes the paper.

Throughout, bold lowercase letters denote vectors and bold uppercase letters denote matrices. The norm $\| \cdot \|$ denotes the Euclidean norm and its induced operator norm. The symbol $\Dd$ denotes the Fréchet derivative~\cite[Definition 17.1]{bauschke2017convex}, while $\nabla$ and $\nabla^2$ denote the gradient and Hessian, respectively. Unless specified, derivatives are taken in canonical coordinates, and $\mathrm{Lip}(f)$ denotes the Lipschitz constant of $f$.

\section{Local Geometry of Separable Problems}\label{sec:local-geometry}

\subsection{Problem Formulation}

We consider an observation vector $\bm{z}\in \mathbb{R}^N$ of the form
\begin{equation}\label{eq:observation-model}
\bz = \A(\bx^\star)\by^\star + \bw,
\end{equation}
where $\bx^\star \in \Omega \subset \R^p$ and $\by^\star \in \mathbb{R}^d$ denote the unknown ground-truth nonlinear parameters and mixing coefficients, respectively, and $\bw \in \mathbb{R}^N$ represents additive noise. The feasible set $\Omega$ is convex and compact, and the manifold map $\bx \mapsto \A(\bx)$ is full-rank and three times Fréchet differentiable everywhere in $\Omega$. The induced operator norms of $\A(\cdot)$ and its derivatives are uniformly bounded over $\Omega$.

We study reconstruction of signals of the form~\eqref{eq:observation-model} by way of the \emph{least squares estimator}, defined as
\begin{equation}\label{eq:lstq}
    \hat{\btheta} \in \argmin_{\bx \in \Omega, \by \in \R^{d}} \L(\btheta) = \tfrac{1}{2} \|\bz - \A(\bx)\by\|^2, \quad \btheta \coloneq (\bx, \by).
\end{equation}
We define the residual $\br : \Omega \times \R^{p + d} \to \R^N$ with
\begin{equation}\label{eq:residual}
    \br(\btheta) = \bz - \A(\bx)\by.
\end{equation}

The matrix $\A(\bx)$ depends nonlinearly on $\bx$, meaning the optimization problem~\eqref{eq:lstq} is generally non-convex. We therefore study the \emph{local geometry} of the loss in a neighborhood of the ground truth parameter $\btheta^\star$, defined by Eq.~\eqref{eq:observation-model}. The ground truth satisfies 
\begin{equation*} 
    \br(\btheta^\star)=\bw.
\end{equation*}

The behavior of any gradient-based optimization method near a solution is entirely governed by the local curvature of the loss, \emph{i.e.}, by its Hessian. Consequently, a purely geometric characterization of the loss around $\btheta^\star$ yields convergence guarantees that are independent of the specific algorithm chosen by a practitioner (\emph{e.g.}, gradient descent, Gauss--Newton, or quasi-Newton methods).
This favorable local geometry is formalized through the notion of a strong basin of attraction.

\begin{definition}[Strong basin of attraction]
    Suppose the loss $\L:~\Omega \times \R^{d} \to \R_+$ of the non-convex optimization problem~\eqref{eq:lstq} is twice continuously differentiable with respect to $\btheta$, and that $\btheta^\star$ is the ground truth.
    A neighborhood $\mathcal{N}$ of $\btheta^\star$ is called a \emph{strong basin of attraction} if there exist constants $0 < \alpha \leq \beta < \infty$ such that $\L$ is $\alpha$-strongly convex and $\beta$-smooth on $\mathcal{N}$. That is, for all $\btheta \in \mathcal{N}$ and all directions $\bu \in \R^{p + d}$,
    \begin{equation}\label{eq:alpha-beta-geometry}
        \alpha \norm{\bu}^2 \leq \left\langle \nabla^2\L(\btheta)\bu, \bu \right\rangle \leq \beta \norm{\bu}^2.
    \end{equation}
\end{definition}

Within a strong basin of attraction, the Hessian is uniformly strongly coercive and uniformly bounded above. As a consequence, the loss behaves locally like a well-conditioned quadratic objective.
This ensures stable local identifiability and linear convergence of standard gradient-based methods from any initialization inside the basin~\cite[Theorem 3.3]{Nocedal2006}.

The radius of the largest ball centered at $\btheta^\star$ that is fully contained in such a region is called the \emph{radius of $\alpha$-convexity}. If $\alpha=0$, it is called the \emph{radius of strong convexity}. It quantifies how far one may perturb the true parameters while preserving this favorable local geometry. This concept is visualized in figure~\ref{fig:basin-visualization}, in terms of the distance $\rho$ introduced in the next section. The top panel shows a multi-modal optimization landscape typical in the non-convex setting. The bottom panel shows the regions of $\alpha$-convexity and convexity as a function of $\rho$ characterized by $\lambda_{\min}(\H(\btheta))$.

\subsection{Strong Basin of Attraction of Least Squares for Separable Problems}

In this section, we harness the well-behavedness of the manifold map $\A(\cdot)$ over $\Omega$ to produce a lower bound on the radius of strong convexity for the least squares estimator. We express the radius in terms of global spectral properties of $ \A(\cdot)$ and its Fréchet derivatives. Furthermore, we establish stability under noise for the estimate $\hat{\btheta}$ that can be recovered by any standard gradient-based method.

The lower bound on the radius is obtained by quantifying Hessian spectrum distortion around the ground truth in terms of parameter perturbations. Under the assumed Lipschitz regularity of the forward map, spectrum distortions are characterized by the Lipschitz constants of $\A(\cdot)$ and its derivatives, herein defined via the global spectral constants
\begin{equation}\label{eq:spectral-norms}
    \sigma_k \coloneq \sup_{\bx \in \Omega} \| \Dd^k \A(\bx) \|,
\end{equation}
where $\Dd^k \A(\bx)$ denotes the $k$-th order Fréchet derivative of $\bx \mapsto \A(\bx)$ evaluated at $\bx$, for $k \in \{0,1,2,3\}$.

The different nature, or physical meaning, of nonlinear parameters $\bx$ and linear parameters $\by$ implies the Euclidean distance over $\Omega \times \R^{d}$ is of ambiguous meaning. For this reason, we define the \emph{unmixing metric} 
\begin{equation}\label{eq:rho}
    \rho(\btheta, \btheta^\star) \coloneq (\sigma_2 \norm{\by^\star}  + \sigma_1)\| \bx - \bx^\star \| + \sigma_1 \| \by - \by^\star \|,
\end{equation}
and observe that the Jacobian of the residual~\eqref{eq:residual} satisfies
\begin{equation*} 
    \|\bm{J}(\btheta)-\bm{J}(\btheta^\star)\|
         \leq \rho(\btheta, \btheta^\star).
\end{equation*}
A detailed derivation of the above inequality can be found in Appendix~\ref{app:jacobian-perturbation}.
    
By expressing distances in the parameters as a function of the induced Jacobian perturbation, the unmixing metric naturally accounts for motion over the parameter space $\Omega \times \R^{d}$ in terms of changes in curvature of the loss, rendering it a natural metric to use when studying separable problems.

For notational convenience, define $\H(\btheta) \coloneq \nabla^2 \L(\btheta)$. Distortions of the spectrum of $\H(\btheta)$ around the ground truth $\btheta^\star$ are controlled using Weyl's inequality~\cite[Corollary 8.1.6]{golub2013MatrixComputations}, which, for eigenvalue index $\ell \in \{1, \dots, p + d\}$, reads
\begin{equation}\label{eq:weyl}
    |\lambda_\ell(\H(\btheta)) - \lambda_\ell(\H(\btheta^\star))| \leq \| \H(\btheta) - \H(\btheta^\star) \|.
\end{equation}
 The \emph{Weyl envelope} is the right-hand side of this inequality. The maximum radius $\rho$ such that the Weyl envelope is bounded above by $\lambda_{\min} (\H(\btheta^\star)) - \alpha$ is a lower bound on the radius of $\alpha$-convexity, by Inequality~\eqref{eq:alpha-beta-geometry}. 

 \begin{figure}[t]
     \centering
     \includegraphics[width=\linewidth]{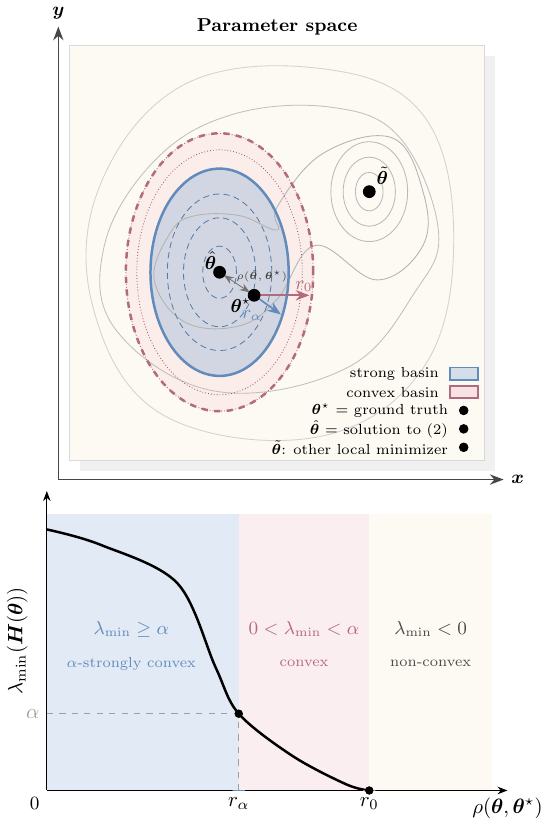}
     \caption{Visualization of the optimization geometry of a separable problem. Top: optimization landscape, $\alpha$-strongly convex region in blue, convex region in pink. Bottom: landscape in terms of the spectrum of $\H(\btheta)$.}
     \label{fig:basin-visualization}
 \end{figure}
 
For any two $\btheta, \btheta^\star \in \Omega$, the Weyl envelope expands into a second-degree polynomial on $\rho(\btheta, \btheta^\star)$, which in turn enables the derivation of the desired bound on the radius of strong convexity. This result is summarized in the following theorem:

\begin{theorem}[Basin of attraction of LS for separable problems]\label{thm:least-squares-basin}
    For any $\btheta \in \Omega$ the Weyl envelope of the Hessian satisfies
    \begin{equation*}
        \| \H(\btheta) - \H(\btheta^\star) \| \leq c_1(\btheta^\star)\rho(\btheta,\btheta^\star) + c_2(\btheta^\star)\rho(\btheta,\btheta^\star)^2.
    \end{equation*}
    with 
    \begin{align*}
        c_1(\btheta^\star) &\coloneq 2\sigma_1\|\by^\star\| + 2\sigma_0 + c_{r,0}(\btheta^\star)c_{r,1}(\btheta^\star) + \|\bw\|c_{r,2}(\btheta^\star),  \\
        c_2(\btheta^\star) &\coloneq 1 + c_{r,1}(\btheta^\star)c_{r,2}(\btheta^\star),
    \end{align*}
    and constants $c_{r,0}, c_{r,1}$ and $c_{r,2}$ defined as
    \begin{subequations}\label{eq:equiv_constants}
        \begin{align}
            c_{r,0}(\btheta^\star) &\coloneq \sigma_2\|\by^\star\| + 2\sigma_1, \\
            c_{r,1}(\btheta^\star) &\coloneq \max \left\{ \frac{\sigma_0}{\sigma_2 \| \by^\star \|+ \sigma_1 } , \frac{1}{\| \by^\star \|} \right\}, \\
            c_{r,2}(\btheta^\star) &\coloneq \max \left\{ \frac{\sigma_3\| \by^\star \| + 2\sigma_2}{\sigma_2 \| \by^\star\| + \sigma_1 } , \frac{2\sigma_2}{\| \by^\star \|} \right\}. 
        \end{align}
    \end{subequations}

    As a consequence, we obtain a lower bound on the radius of  $\alpha$-convexity
    \begin{equation*} 
        r_{\alpha, \mathrm{ls}} \coloneq \dfrac{\sqrt{c_1(\btheta^\star)^2 + 4c_2(\btheta^\star)(\lambda_{\min}(\H(\btheta^\star)) - \alpha)} - c_1(\btheta^\star)}{2c_2(\btheta^\star)}.
    \end{equation*}
\end{theorem}

\begin{proof}
    The proof proceeds by bounding the perturbation of the Hessian
    $\H(\btheta)$ around $\btheta^\star$ in the unmixing metric, and by applying Weyl's inequality.
    
    We decompose the Hessian as the sum of a curvature term $\bm{H}_c(\btheta)$ and a residual term $\bm{H}_r(\btheta)$ so that
    \begin{align*}
        \H(\btheta)
        &= \bm{J}(\btheta)^\top \bm{J}(\btheta)
        + \sum_{\ell=1}^N r_\ell(\btheta)\nabla^2 r_\ell(\btheta) \\
        &\eqcolon \H_c(\btheta) + \H_r(\btheta).
    \end{align*}
    With the triangle inequality, one has $\| \H(\btheta) - \H(\btheta^\star) \| \leq \| \H_c(\btheta) - \H_c(\btheta^\star) \| + \| \H_r(\btheta) - \H_r(\btheta^\star) \|$, and we bound the perturbation separately on the curvature and residual terms.
    
    \noindent \emph{Curvature term.}
    Using standard matrix identities, one obtains
    \begin{multline}
        \|\H_c(\btheta) - \H_c(\btheta^\star)\| \\
        \le
        \big(2\|\bm{J}(\btheta^\star)\|
        + \|\bm{J}(\btheta)-\bm{J}(\btheta^\star)\|\big) 
        \cdot
        \|\bm{J}(\btheta)-\bm{J}(\btheta^\star)\|.
        \label{eq:curvature-perturbation}
    \end{multline}
    and observe
    \begin{equation*} 
        \|\bm{J}(\btheta^\star)\| = \|[(\Dd\A(\bx^\star))\by^\star, \A(\bx^\star)]\| \leq \sigma_1\|\by^\star\| + \sigma_0.
    \end{equation*}
    
    \noindent \emph{Residual term.}
    Define the auxiliary metrics
    \begin{align*}
        \rho_1(\btheta,\btheta^\star)
        &\coloneq \sigma_1\|\by^\star\|\|\bx-\bx^\star\|
        + \sigma_0\|\by-\by^\star\|, \\
        \rho_2(\btheta,\btheta^\star)
        &\coloneq (\sigma_3\|\by^\star\|+\sigma_2)\|\bx-\bx^\star\|
        + 2\sigma_2\|\by-\by^\star\|.
    \end{align*}
    Then the residual Hessian perturbation satisfies (Appendix~\ref{app:residual-perturbation})
    \begin{align}
        \|\H_r(\btheta)-\H_r(\btheta^\star)\|
        &\le c_{r,0}(\btheta^\star)\rho_1(\btheta,\btheta^\star) \nonumber\\
        & + \big(\|\bw\| + \rho_1(\btheta,\btheta^\star)\big)
        \rho_2(\btheta,\btheta^\star).
        \label{eq:residual-perturbation}
    \end{align}

    Since $\rho$, $\rho_1$, and $\rho_2$ define norms on $\Omega\times\R^{d}$, by norm equivalences, it can be verified that the constants in~\eqref{eq:equiv_constants} verify
        \begin{align*}
        \rho_1(\btheta,\btheta^\star)
        \leq c_{r,1}(\btheta^\star)\rho(\btheta,\btheta^\star),\
        \rho_2(\btheta,\btheta^\star)
        \leq c_{r,2}(\btheta^\star)\rho(\btheta,\btheta^\star).
    \end{align*}
    
    Combining~\eqref{eq:curvature-perturbation},~\eqref{eq:residual-perturbation} and~\eqref{eq:equiv_constants} yields the stated Hessian perturbation bound. The expression for $r_{\alpha, \mathrm{ls}}$ follows by applying Weyl's inequality and solving the resulting quadratic inequality.
\end{proof}

\begin{corollary}[Stability of least squares recovery]\label{cor:ls-stability}
    The local minimizer of the loss within the strong basin of attraction $\hat \btheta$ satisfies
    \begin{align*} 
        \rho(\hat{\btheta}, \btheta^\star) \leq (\sigma_2\|\by^\star\| + \sigma_1)\dfrac{\|\bm{J}(\btheta^\star)^\top \bw\|}{\alpha}.
    \end{align*}
\end{corollary}

\begin{proof}
    In the strong basin of attraction, the Hessian of $\L$ is $\alpha$-strongly coercive. Therefore, its gradient is $\alpha$-strongly monotone, so that
    \begin{align*}
        \alpha \|\hat{\btheta} - \btheta^\star \|^2 &\leq \langle \nabla\L(\hat{\btheta}) - \nabla\L(\btheta^\star), \hat{\btheta} - \btheta^\star\rangle \\
        &\leq \| \nabla\L(\btheta^\star) \| \|\hat{\btheta} - \btheta^\star \| \leq \| \bm{J}(\btheta^\star)^\top \bw \| \|\hat{\btheta} - \btheta^\star \|,
    \end{align*}
    yielding the result.
\end{proof}

\subsection*{Discussion}

The derivation in Theorem~\ref{thm:least-squares-basin} is carefully laid out to yield a bound depending explicitly on the finer regularity properties of the forward map, characterized by the spectral constants $\sigma_k$. For this reason, the result makes explicit which properties of $\A(\cdot)$ result in increased initialization robustness; large radii of strong convexity arise when the manifold map exhibits limited higher-order variation, as captured by the spectral constants $\sigma_2$ and $\sigma_3$. For such problems, the Hessian
varies slowly around the ground truth, and the least squares objective remains well-conditioned over a larger neighborhood. 

As a consequence, the sharpness of the bound is controlled by the uncertainty on the nonlinear parameters $\bx$, expressed through the size of the feasible set $\Omega$. Moreover, our derivation hinges on applying Weyl's inequality on the minimum eigenvalue, which is pessimistic for ill-conditioned matrices. However, the resulting bound applies to a broad class of separable problems. Tighter guarantees generally require exploiting additional problem-specific structure, such as special kernel properties, orthogonality relations, or sparsity patterns. The present result thus provides a general baseline guarantee, upon which sharper results can be built when additional structure is available.

Furthermore, the reconstruction error is governed by the smoothness and conditioning of the Jacobian. In the strong basin of attraction, the local strong convexity constant satisfies
\begin{equation*}
    \alpha \approx \lambda_{\min}(\bm{J}(\btheta^\star)^\top \bm{J}(\btheta^\star))
    = \sigma_{\min}(\bm{J}(\btheta^\star))^2.
\end{equation*}
Moreover, $\sigma_2\|\by^\star\| + \sigma_1$ provides an upper bound on the Lipschitz constant of $\bm{J}(\cdot)$ in the standard $\ell_2$ topology. Consequently, the recovery error satisfies
\begin{equation*}
    \rho(\hat{\btheta}, \btheta^\star) \lesssim 
    \frac{\mathrm{Lip}^{\ell_2}(\bm{J}(\cdot))
    \|\bm{J}(\btheta^\star)\|^2}{\sigma_{\min}(\bm{J}(\btheta^\star))^4} \| \bw \|,
\end{equation*}
so that the best-case reconstruction accuracy is achieved when the Jacobian varies slowly and remains well conditioned at the ground truth.

\subsection{Variable Projection}

The variable projection method~\cite{variableProjection} exploits the separable structure of the model by eliminating the linear variables from the optimization problem~\eqref{eq:lstq}.
In our setting, the matrix $\A(\bx)$ has full column rank for all $\bx \in \Omega$. Consequently, for any fixed $\bx$, the problem
\begin{equation*}
    \argmin_{\by \in \R^{d}} \tfrac{1}{2}\|\bz - \A(\bx)\by\|^2
\end{equation*}
admits a unique solution, given in closed form by
\begin{equation*}
    \hat{\by}(\bx) = \A(\bx)^+ \bz,
\end{equation*}
where 
\begin{equation*}
 \A(\bx)^+ = \left(\A(\bx)^\top \A(\bx)\right)^{-1}\A(\bx)^\top
\end{equation*}
denotes the Moore--Penrose pseudo-inverse. The optimal coefficient vector $\hat{\by}(\bx)$ implicitly defines a \emph{lifting map} that embeds the reduced problem back into the full parameter space
\(
    \btheta(\bx) \coloneq (\bx, \hat{\by}(\bx))
\).
Substituting the optimal coefficients into the loss yields the \emph{projected} objective
\begin{equation*}
    \L_\mathrm{vp}(\bx) = \L(\btheta(\bx)) = \tfrac{1}{2}\|(\Id - \A(\bx)\A(\bx)^+)\bz\|^2,
\end{equation*}
and, as a consequence, the \emph{projected least squares estimator}, given by
\begin{equation}\label{eq:projected-lstq}
    \hat{\bx} \in \argmin_{\bx \in \Omega} \L_\mathrm{vp}(\bx).
\end{equation}
This results in an optimization landscape expressed entirely in terms of $\bx$, and therefore a sizable reduction in degrees of freedom.

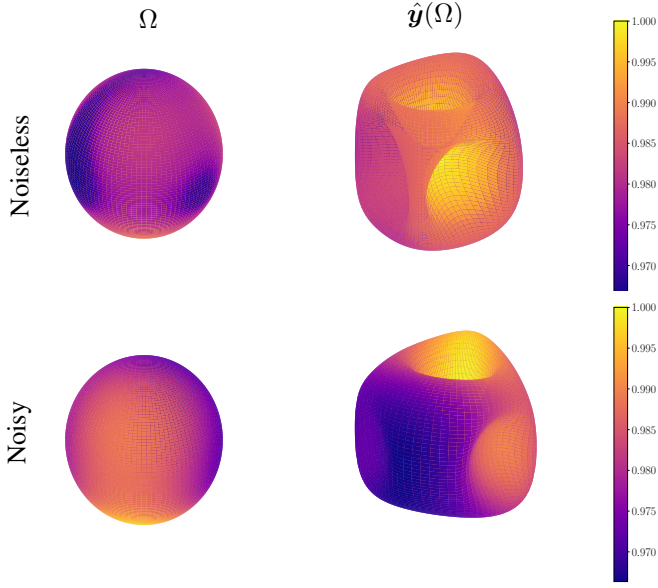
\begin{figure}[t]
    \centering
    \begin{tikzpicture}[node distance=1.0em]
    
    \node (top) {\includegraphics[width=\linewidth]{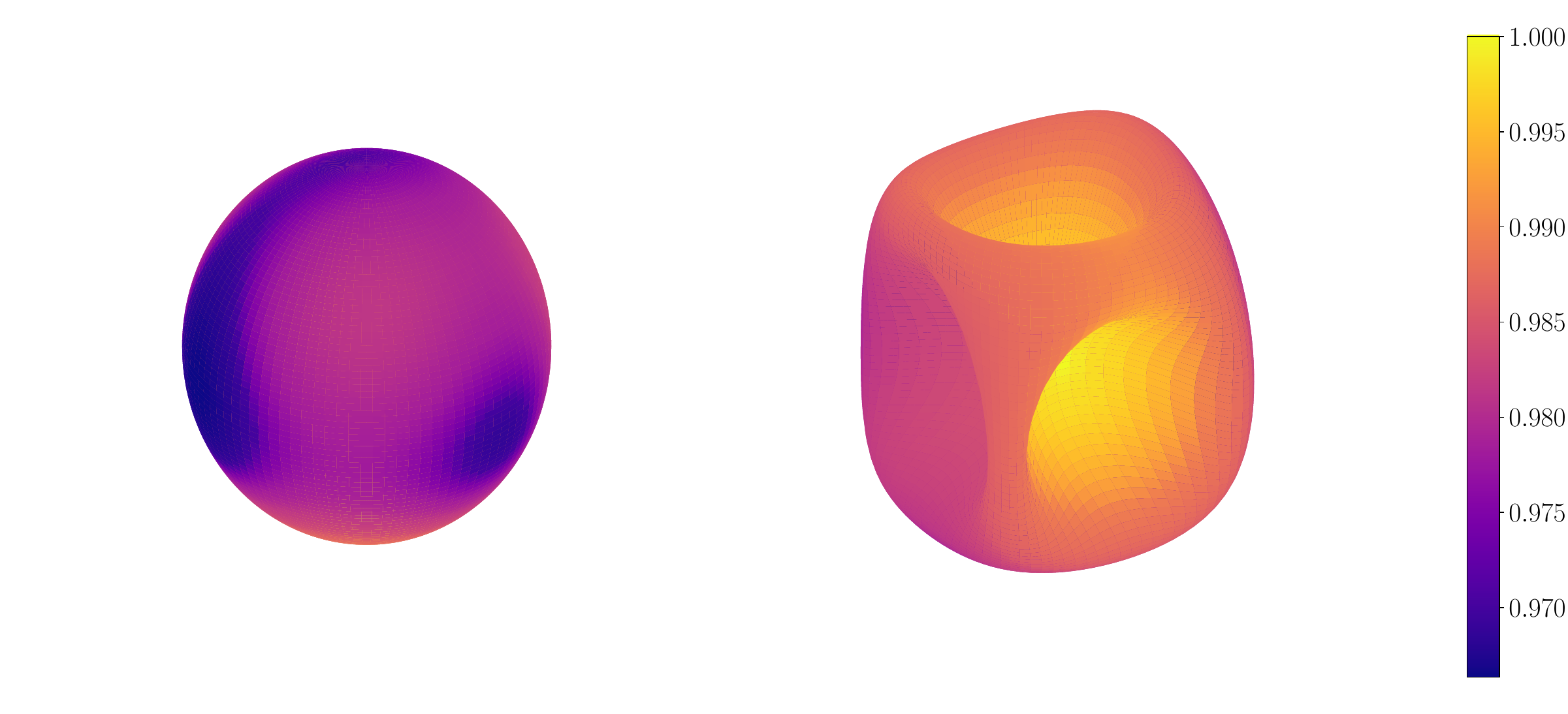}};
    
    \node[below=-0.5cm of top] (bottom) {\includegraphics[width=\linewidth]{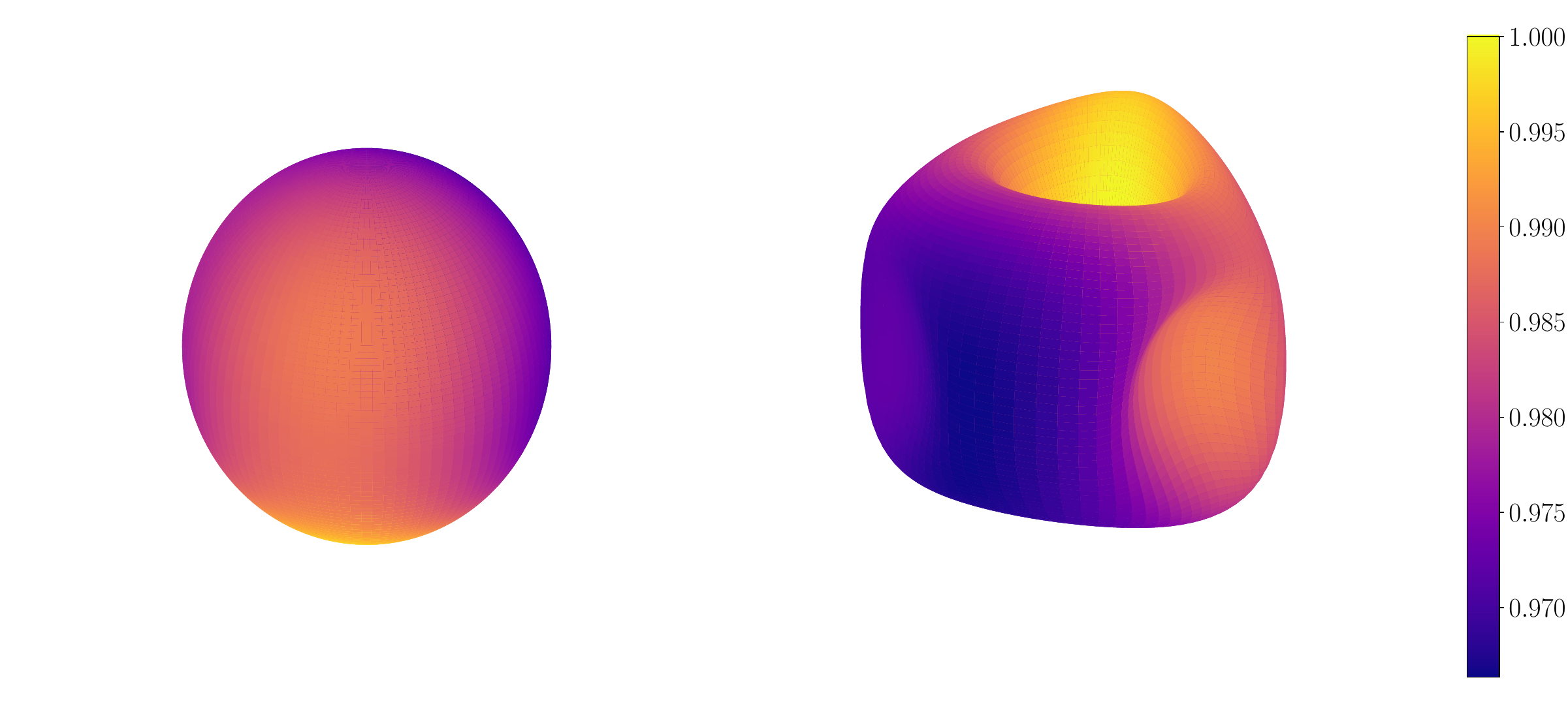}};
    
    \node[above=of top, xshift=-2.3cm, yshift=-0.9cm] {$\Omega$};
    \node[above=of top, xshift=1.5cm, yshift=-0.9cm] {$\hat{\by}(\Omega)$};
    
    \node[left=of top, rotate=90, yshift=-0.9cm, xshift=0.7cm] {Noiseless};
    \node[left=of bottom, rotate=90, yshift=-0.9cm, xshift=0.7cm] {Noisy};
    
    \end{tikzpicture}
    \caption{Geometric effect of variable projection on the loss landscape. The feasible set $\Omega$ for the nonlinear parameters (left) and its image under the optimal value map $\hat{\by}(\cdot)$ (right) are shown for noiseless (top) and noisy (bottom) observations.}
    \label{fig:var-pro-feasible-sets}
\end{figure}

The effect of variable projection is illustrated in Fig.~\ref{fig:var-pro-feasible-sets}. The figure shows the optimal value surface defined over $\R^{d}$ via the lifting map $\bm{\theta}(\bm{x})$, alongside the feasible set $\Omega$, herein defined as a sphere centered at the nonlinear ground truth parameters $\bx^\star$. The illustration was generated by considering a PSF unmixing problem (formally introduced in Sec.~\ref{sec:psf-unmixing}) with $p=3$ and $d=3$. The top row shows the noiseless setting, while the bottom row corresponds to observations perturbed by additive, Gaussian noise at $-10\,\mathrm{dB}$. Coloring indicates normalized loss values, with histogram matching applied between noiseless and noisy cases to enable qualitative comparison across panels.

The previous comparison highlights the key advantage of variable projection: restricting the optimization to optimal linear parameters eliminates directions associated with linear variables that do not reduce the objective, leading to a substantially better-conditioned loss landscape.

The improvement in conditioning is propagated to the Hessian of the projected objective, given by
\begin{equation}\label{eq:projected-hessian}
    \H_\mathrm{vp}(\bx) \coloneq \nabla^2 \L_\mathrm{vp}(\bx) = \Dd \btheta(\bx)^\top \H(\btheta(\bx)) \Dd \btheta(\bx).
\end{equation}
The curvature of the projected objective then corresponds to the full Hessian restricted to the tangent space of the optimal value surface.
The derivative of the lifting map $\Dd \btheta(\bx)$ takes the form of a truncated Schur complement,
\begin{equation*}
    \Dd \btheta(\bx) =
    \begin{bmatrix}
        \Id_p & - \nabla_{\bx \by} \L(\btheta(\bx))^\top (\nabla_{\by \by}^2 \L(\btheta(\bx)))^{-1}
    \end{bmatrix}^\top,
\end{equation*}
which removes curvature contributions corresponding to linear variables.

The curvature variations of the projected objective are directly controlled by those of the full problem evaluated along the optimal-value surface. Along this surface, the spectral spread of the Hessian is minimized, so Weyl's inequality provides tight control of the minimum eigenvalue. This is the key fact that enables transferring the radius $r_{\alpha, \mathrm{ls}}$ derived in Theorem~\ref{thm:least-squares-basin} to this setting.

\begin{lemma}[Projected Weyl envelope]\label{thm:projected-envelope}
    For all $\ell \in \{1, \dots, p \}$ the eigenvalues of the projected Hessian satisfy
    \begin{equation}\label{eq:projected-envelope}
        |\lambda_\ell(\H_\mathrm{vp}(\bx)) - \lambda_\ell(\H_\mathrm{vp}(\bx^\star))| \leq K \| \H(\btheta(\bx)) - \H(\btheta(\bx^\star)) \|,
    \end{equation}
    with 
    \begin{equation*}
        K \coloneq \ \left(1 + \dfrac{\| \nabla_{\bx \by}^2 \L(\btheta(\bx))\|}{\sigma_{\min}(\A(\bx))} \right)^2.
    \end{equation*}
\end{lemma}

\begin{proof}
    As a notation shorthand, write 
    \begin{align*}
        \Delta\H_{\bx \bx} &\coloneq \nabla_{\bx \bx}^2 \L(\btheta(\bx)) - \nabla_{\bx \bx}^2 \L(\btheta(\bx^\star)), \\
        \Delta\H_{\by \bx} &\coloneq \nabla_{\by \bx}^2 \L(\btheta(\bx)) - \nabla_{\by \bx}^2 \L(\btheta(\bx^\star)), \\
        \Delta\H_{\by \by} &\coloneq \nabla_{\by \by}^2 \L(\btheta(\bx)) - \nabla_{\by \by}^2 \L(\btheta(\bx^\star)), 
    \end{align*}
    and
    \begin{equation*}
        \Delta \H \coloneq \H(\btheta(\bx)) - \H(\btheta(\bx^\star)).
    \end{equation*}
    Using expression~\eqref{eq:projected-hessian}, the triangle inequality and the sub-multiplicativity of $\|\cdot\|$,
    \begin{align}
        \| \H_\mathrm{vp}(\bx) - \H_\mathrm{vp}(\bx^\star) \| &\leq \| \Delta\H_{\bx \bx}\| \nonumber \\
        &\hspace{-80pt}+ 2 \| \nabla_{\bx \by}^2 \L(\btheta(\bx))\|\| \nabla_{\by \by}^2 \L(\btheta(\bx))^{-1}\|\| \Delta\H_{\by \bx}\| \nonumber \\
        &\hspace{-80pt}+ \| \nabla_{\bx \by}^2 \L(\btheta(\bx))\|^2\| \nabla_{\by \by}^2 \L(\btheta(\bx))^{-1}\|^2 \| \Delta\H_{\by \by}\|. \label{proof:schur-expansion}
    \end{align}
    The restriction to a sub-block is a contraction, therefore
    \(
        \| \Delta\H_{\bx \bx}\| \leq \| \Delta \H\| ,\ \| \Delta\H_{\by \bx} \| \leq \| \Delta \H\|, 
    \)
    and
    \(
        \| \Delta\H_{\by \by} \| \leq \| \Delta \H\|.
    \)
    Finally, because $\A(\cdot)$ is full-rank everywhere, $\| \nabla_{\by \by}^2 \L(\btheta(\bx))^{-1}\| = \sigma_{\min}(\A(\bx))^{-1}$.
    Substitution of these into~\eqref{proof:schur-expansion} yields
    \begin{multline*}
        \| \H_\mathrm{vp}(\bx) - \H_\mathrm{vp}(\bx^\star) \| \\
        \hspace{10pt}\leq \left(1 + 2\tfrac{\| \nabla_{\bx \by}^2 \L(\btheta(\bx))\|}{\sigma_{\min}(\A(\bx))} +\tfrac{\| \nabla_{\bx \by}^2 \L(\btheta(\bx))\|^2}{\sigma_{\min}(\A(\bx))^2}\right)\| \Delta \H\|,
    \end{multline*}
    thus concluding the proof.
\end{proof}

The coupling factor $K$ depends on the conditioning of the linear subproblem, and on the correlation between the column spaces spanned by $\A(\bx)$ and $\Dd \A(\bx)$, characterized by the minimum principal angle between the respective singular spaces. In neighborhoods of the ground truth, $K$ is uniformly bounded above, and for well-conditioned problems $K \approx 1$.

The lifting map $\bx \mapsto \btheta(\bx)$ propagates motion along $\Omega$ to variations along the optimal value surface $\hat \by(\Omega)$. The restriction $(\bx, \bx') \mapsto \rho(\btheta(\bx), \btheta(\bx'))$ posits a natural mechanism to measure these variations. Observe that the perturbation in the lifted linear coefficients satisfies
\begin{multline*}
    \|\hat{\by}(\bx) - \hat{\by}(\bx^\star) \| \leq \| \A(\bx)^+\by - \by^\star\| + \|\by^\star - \hat{\by}(\bx^\star)\| \nonumber \\
    \leq  \sigma_1 \| \bx - \bx^\star \| \| \A(\bx)^+\|\|\by^\star \| + \|\A(\bx)^+\bw \| + \|\bw\|. 
\end{multline*}
Where $\sigma_1$, taken as in definition~\eqref{eq:spectral-norms}, is the Lipschitz constant of $\A(\cdot)$. Furthermore, $\A(\bx)$ is full rank over $\Omega$. Therefore, there exists a constant
\begin{equation*}
    \widetilde{\sigma}_{\min} = \min_{\bm{x} \in \Omega} \sigma_{\min}\left( \A(\bm x) \right) > 0,
\end{equation*}
satisfying, for any $\bx$ in $\Omega$, the uniform bound
\begin{equation*}
    \|\A(\bx)^+\| = \sigma_{\min}(\A(\bx))^{-1} \le \widetilde{\sigma}_{\min}^{-1}.
\end{equation*}
As a consequence, 
\begin{equation}\label{eq:feedforward-error}
   \|\hat{\by}(\bx) - \hat{\by}(\bx^\star) \| \leq  \frac{\sigma_1}{\widetilde{\sigma}_{\min}} \| \bx - \bx^\star \|\|\by^\star \| + \left(1 + \frac{1}{\widetilde{\sigma}_{\min}}\right)\|\bw\|.
\end{equation}
Substitution of the above equation into the definition of the unmixing metric $\rho$~\eqref{eq:rho} gives
\begin{equation}\label{eq:rho-lift-bound}
\rho(\btheta(\bx), \btheta(\bx^\star))
\leq c_{\mathrm{vp}}\|\bx-\bx^\star\| + \sigma_1\Bigl(1+\frac{1}{\widetilde{\sigma}_{\min}}\Bigr)\|\bw\|,
\end{equation}
with
\begin{equation*}   c_{\mathrm{vp}}=\bigl(\sigma_2+\sigma_1^2/\widetilde{\sigma}_{\min}\bigr)\|\by^\star\|+\sigma_1.
\end{equation*}

The bound~\eqref{eq:rho-lift-bound} characterizes the propagation of perturbations of $\bx$, through the lifting map, to the full problem parameter $\btheta = (\bx, \by)$. In the noisy case, lifted distances are distorted by up to a factor of $\sigma_1(1 + \widetilde{\sigma}_{\min}^{-1})\|\bw\|$.

The eigenvalues of a Hermitian matrix interlace with the eigenvalues of its Schur complement~\cite{Smith1992InterlacingSchur}. This implies the existence of $k_\mathrm{vp} \geq 1$ such that
\begin{equation*}
    \lambda_{\min}(\H_\mathrm{vp}(\bx^\star)) = k_\mathrm{vp} \lambda_{\min}(\H(\btheta^\star)).
\end{equation*}

Combining the ambient strong-convexity radius $r_{\mathrm{LS}}$ from Theorem~\ref{thm:least-squares-basin} with the projected eigenvalue perturbation bound of Lemma~\ref{thm:projected-envelope} and the bound on the lifted linear parameter perturbation~\eqref{eq:rho-lift-bound} enables the translation of the ambient basin guarantee into a guarantee for the projected objective. While the radius can be transferred in the ill-conditioned ($K >> 1$) and noisy cases, we highlight the following structural result:

\begin{theorem}[Radii comparison inequality]\label{thm:radii-comparison}
    Suppose 
    \begin{equation*}
        \nabla_{\bx \by}^2 \L(\btheta(\bx)) = 0 \text{ on } \Omega
    \end{equation*}
    and measurements are noiseless ($\bw = 0$, a.s.). In the topology induced by $\rho$, the radius of strong convexity for the projected objective admits the comparison
    \begin{equation*} 
    \sqrt{k_{\mathrm{vp}}}\frac{r_{\mathrm{LS}}}{c_\mathrm{vp}} \leq r_{\mathrm{vp}}
        \leq k_{\mathrm{vp}}\frac{r_{\mathrm{LS}}}{c_\mathrm{vp}}.
    \end{equation*}
\end{theorem}

\begin{proof}
    The proof proceeds by algebraic manipulations; details can be found in Appendix~\ref{app:vp-radius}.
\end{proof}

\begin{lemma}[Stability of Variable Projection Recovery]\label{cor:vp-stability}
    Denote $\alpha_{\mathrm{vp}}$ the coercivity constant of the projected Hessian. The local minimizer of the loss~\eqref{eq:projected-lstq} within the strong basin of attraction $\hat \bx$ satisfies
    \begin{multline}\label{eq:vp-stability}
        \rho(\hat{\btheta}_\mathrm{vp}, \btheta^\star) \leq (\sigma_2\|\by^\star\| + \sigma_1)\frac{\|\bm{J}_\mathrm{vp}(\bx^\star)^\top \bw\|}{\alpha_{\mathrm{vp}}}  \nonumber \\
        + \frac{\sigma_1}{\widetilde{\sigma}_{\min}} \|\bw\| + \frac{\sigma_1^2}{\widetilde{\sigma}_{\min}} \frac{\|\bm{J}_\mathrm{vp}(\bx^\star)^\top \bw\|}{\alpha_{\mathrm{vp}}} \|\by^\star\|.
    \end{multline}
\end{lemma}

\begin{proof}
    Strong monotonicity of the gradient of the projected loss gives
    \begin{equation*}
        \|\hat{\bx} - \bx^\star \| \leq \frac{\|\bm{J}_\mathrm{vp}(\bx^\star)\top \bw\|}{\alpha_{\mathrm{vp}}}. 
    \end{equation*}
    Furthermore, from~\eqref{eq:feedforward-error} we obtain
    \begin{equation*}
        \|\hat{\by}(\bx) - \by^\star \| \leq  \frac{\sigma_1}{\widetilde{\sigma}_{\min}} \| \bx - \bx^\star \|\|\by^\star \| + \frac{1}{\widetilde{\sigma}_{\min}} \|\bw \|.
    \end{equation*}
    Combining the two previous equations into the definition of $\rho(\cdot, \cdot)$ gives the result.
\end{proof}

\subsection*{Discussion}

An alternative derivation of the radius of strong convexity can be constructed by bounding the Weyl envelope of~\eqref{eq:projected-hessian}. However, this envelope is a high degree polynomial in $\|\bx - \bx^\star\|$, severely hindering its interpretability. In contrast, the derivations leading up Theorem~\ref{thm:radii-comparison} are carefully constructed to expose the mechanism by which variable projection operates: restriction to optimal directions. This restriction leads to significantly better conditioned Jacobian and Hessian matrices, thus providing, in addition to improved numerical performance, tighter theoretical guarantees.

\section{Application to the PSF Unmixing Problem}\label{sec:psf-unmixing}
The generic theory developed above applies to any separable problem satisfying the stated regularity conditions. We now instantiate it for PSF unmixing, where the structure of $\A(\bx)$ enables the expression of the resulting bounds in terms of physically interpretable quantities.

Point spread function (PSF), or kernel, unmixing forms an important class of signal recovery problems with separable structure. In this setting, each column of the matrix $\A(\bx)$ is obtained by evaluating a parameterized PSF at shifted locations on a fixed sampling grid. Concretely, let
\begin{equation*}
    \mathcal{T} = \{t_{i,k}\} \subset [-T/2,T/2] ,\quad T > 0,
\end{equation*}
denote a dictionary of spike locations, comprising $p \geq 1$ groups of $q \geq 1$ spikes with no repeated locations, and let
\begin{equation*}
    g(\cdot, \cdot) : X \times \R \to \R
\end{equation*}
be a class of PSFs parameterized by the shape parameter $x \in X = [x_{\min}, x_{\max}]$. Given a sampling grid $\bm{t} \in \R^N$, each location $t_{i,k} \in \mathcal{T}$ generates a column of $\A(\bx)$ via
\begin{equation*}
    \ba_{i,k}(x_i) = g(x_i, \bm{t} - t_{i,k}),
\end{equation*}
where subtraction and evaluation are performed point-wise. We then define the matrix-valued map $\A : \Omega \to \R^{N \times pq}$, with $\Omega = X^p$, as
{\small
    \begin{equation*}
        \A(\bx)
        =
        [
        \ba_{1,1}(x_1), \dots, \ba_{1,q}(x_1),
        \dots,
        \ba_{p,1}(x_p), \dots, \ba_{p,q}(x_p)
        ].
    \end{equation*}
}
We thus pose PSF unmixing as the problem of finding shape parameters $\bx^\star$ and mixing coefficients (spike weights) $\by^\star$ from an observation $\bz$ obtained via~\eqref{eq:observation-model}.

A central structural assumption in PSF unmixing is that \emph{all spike blocks are active}. Meaning that for each spike group, at least one spike is present in the observation $\bz$, thus ensuring identifiability of all latent parameters. Furthermore, we suppose the number of samples $N$ and the parameter domain $X$ are specified so that $\A(\bx)$ is full rank everywhere on $\Omega$. Finally, we assume that for each fixed $t \in [-T/2,T/2]$, the function $x \mapsto g(x,t)$ is three times continuously differentiable, and that for $ k \in \{0,1,2,3\}$ the derivatives satisfy
\begin{equation*}
    \tfrac{\partial^k}{\partial x^k} g(x,\cdot)
    \in L^2(\R, \diff t).
\end{equation*}

Under these conditions, the theory developed in the previous section transfers to the PSF unmixing setting: Theorems~\ref{thm:least-squares-basin} and~\ref{thm:radii-comparison}, alongside Corollaries~\ref{cor:ls-stability} and~\ref{cor:vp-stability} apply.

\subsection{Simplification of Spectral Constants}

The column-block structure of $\A(\cdot)$ implies the higher-order Fréchet derivatives become diagonal. Therefore, finding each spectral constant $\sigma_k$~\eqref{eq:spectral-norms} collapses to at most $p$ one-dimensional, completely decoupled problems. The following lemma states this fact:

\begin{lemma}[Global operator norms]\label{thm:sigmas}
    For $i \in \{1, \dots, p\}$, and $k \in \{0,1,2,3\}$ define the column block derivative
    \begin{equation*}
        \partial_k \A_{i}(x_i) \coloneq \left[\tfrac{\partial^k}{\partial x_i^k} \ba_{i, 1}(x_i), \dots, \tfrac{\partial^k}{\partial x_i^k} \ba_{i, q}(x_i) \right] \in \R^{N \times q}.
    \end{equation*}
    The constants $\sigma_k$, defined in~\eqref{eq:spectral-norms}, satisfy
    \begin{equation*}
        \sigma_k \leq \begin{cases}
            \sqrt{p}\max\limits_{i \in \{1, \dots, p\}} \sup\limits_{x_i \in X} \| \A_{i}(x_i) \| & k = 0, \\
            \max\limits_{i \in \{1, \dots, p\}} \sup\limits_{x_i \in X} \| \partial_k \A_{i}(x_i) \| & \text{ else}.
        \end{cases}
    \end{equation*}
\end{lemma}

\begin{proof}
    For the $k=0$ case, we simply write out, for any $\bx \in \Omega$,
    \begin{equation*}
        \| \A(\bx) \| = \sup_{\| \by \| = 1} \| \A(\bx)\by  \| \leq \sup_{\| \by \| = 1} \sqrt{ \sum_{i=1}^p \| \A_i(x_i)\by_i  \|^2 },
    \end{equation*}
    where $\by_i \in \R^q$ is the $i$-th block of $\by$. Taking the supremum over $\Omega$ gives the result.
    For $k \geq 1$,  the key observation is that the cross derivatives $\frac{\partial^k}{\partial x_j^k} \A_i(x_i)$ vanish. 
    Let $\bu \in \R^p$. Observe the directional derivative $\Dd \A(\bx)[\bu]$ is an element of $\R^{N \times pq}$. The constant $\sigma_1$ then satisfies
    \begin{align*}
        \sigma_1 &= \| \Dd \A(\bx) \| = \sup_{\| \bu \| = 1} \| \Dd \A(\bx)[\bu] \| \\
                               &= \sup_{\| \bu \| = 1}  \sup_{\| \by \| = 1} \| \Dd \A(\bx)[\bu]\by \| \\
                               &= \sup\limits_{
                            \begin{array}{c}
                                    \| \bu \| = 1  \\
                                    \| \by \| = 1
                               \end{array}} \left\| \sum_{i=1}^p u_i \partial \A_i(x_i)\by_i \right\| \\
                               &\leq \sup\limits_{ \begin{array}{c}
                                    \| \bu \| = 1  \\
                                    \| \by \| = 1
                               \end{array}}  \max_{i \in \{1, \dots, p\}} \| \partial \A_i(x_i)\| \sum_{i=1}^p |u_i| \| \by_i \| \\
                               &\leq \max_{i \in \{1, \dots, p\}} \| \partial \A_i(x_i)\|. 
    \end{align*}
    The cases $k = 2$ and $ k=3$ are dealt with by analogous reasoning.
\end{proof}

\subsection{Coherence}

\begin{figure}[t]
    \centering
    \includegraphics[width=\linewidth]{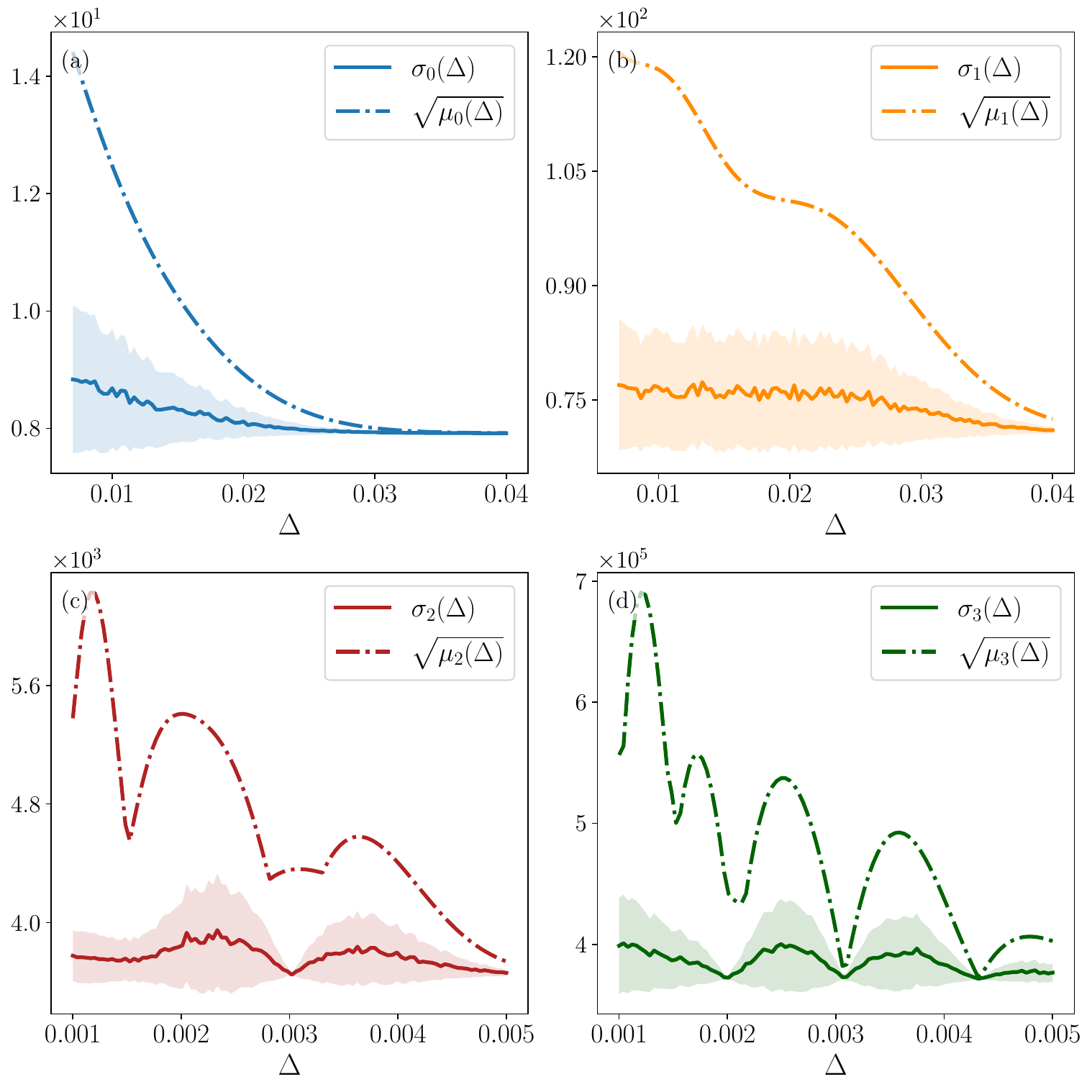}
    \caption{Empirical validation of Theorem~\ref{thm:coherence} for a Gaussian kernel. (a) $k=0$, (b) $k=1$, (c) $k=2$, (d) and $k=3$. Solid curves show the empirical mean of the spectral constants $\sigma_k$ over random dictionaries with minimal separation $\Delta$, with shaded regions indicating one standard deviation. Dotted curves depict the coherence envelope $\mu_k(\Delta)$.}
    \label{fig:coherence}
\end{figure}

In this setting, the size of the strong basin of attraction, and the recovery error are characterized by the \emph{minimal separation} of the support, defined as
\begin{equation} \label{eq:delta}
    \Delta \coloneqq \min_{t,t' \in \T, t \neq t'} |t - t'|.
\end{equation}
When the dictionary $\bm{A}(\bm{x})$ is fixed, the minimal separation~\eqref{eq:delta} 
emerges as a fundamental quantity governing both the feasibility and stability of the 
inverse problem. Its role has been thoroughly investigated in the context of spike 
deconvolution~\cite{chi2020harnessing,ferreira2020stable,hockmann2023weak}, and 
subsequently extended to a broader class of continuous measurement 
operators~\cite{poon2023GeometryOfftheGrida}.

Herein, we introduce a notion of \emph{coherence} that quantifies the similarity 
between kernels and their derivatives when evaluated on a support separated by at 
least $\Delta$. For $k \in \{0,1,2,3\}$, the \emph{$\Delta$-separated correlation} 
of the $k$th kernel derivative is defined as
\begin{equation*}
    \varrho_k(x, \Delta) \coloneq \sup\limits_{|\delta| \geq \Delta} 
    \left|\left\langle \tfrac{\partial^k}{\partial x^k} g(x, \bm{t}),\, 
    \tfrac{\partial^k}{\partial x^k}g(x, \bm{t} - \delta) \right\rangle\right|;
\end{equation*}
For a support separation $\Delta$, the \emph{coherence} of the $k$th kernel 
derivative $\mu_k(\Delta)$ is then given by
\begin{equation*}
    \mu_k(\Delta) \coloneq \sup_{x \in X} \sum_{m \in \Z} \varrho_k(x, |m|\Delta). 
\end{equation*}
In this setting, the abstract spectral constants~\eqref{eq:spectral-norms}, which define the radius of strong convexity, reduce to the dictionary's coherence. In particular, we have the following theorem:
\begin{theorem}\label{thm:coherence}
    Let the minimal separation of the support $\Delta > 0$. The spectral constants~\eqref{eq:spectral-norms} satisfy,
    \begin{equation*}
        \sigma_k \leq \begin{cases}
            \sqrt{p\mu_k(\Delta)} & k = 0, \\
            \sqrt{\mu_k(\Delta)}  & k \in \{1,2,3\}.
        \end{cases}
    \end{equation*}
\end{theorem}

\begin{proof}

From Lemma~\ref{thm:sigmas}, we bound the operator norms of the column blocks
\begin{equation*}
   \| \partial_k\A_i(x_i) \| = \sqrt{\lambda_{\max}(\partial_k\A_i(x_i)^\top\partial_k\A_i(x_i))}.
\end{equation*}
Therefore, we seek an upper bound on the spectral radius of the associated
Gram matrix that holds for all admissible location dictionaries. The Gramian entries are correlations between shifted kernel derivatives. For $m,n \in \{1,\dots,d\}$, the $n$-th entry of the $m$-th row writes
\begin{equation*}
    \big\langle\tfrac{\partial^k}{\partial x_i^k} g(x_i, \bm{t} - t_{i, n}), \tfrac{\partial^k}{\partial x_i^k}g(x_i, \bm{t} - t_{i, m})\big\rangle.
\end{equation*}
These correlations are characterized by the relative distance
\begin{equation*}
    \Delta_{n,m} \coloneq |t_{i,n}-t_{i,m}|.
\end{equation*}
Suppose, w.l.o.g., the locations are organized in ascending order. By minimal separation, 
\begin{equation*}
    (|n - m| - 1)\Delta \le \Delta_{n,m} \le (|n - m| + 1)\Delta.
\end{equation*}
The preceding relation holds independently of the block index $i \in \{1, \dots, p\}$. As a consequence, applying Gershgorin's theorem~\cite[Theorem 8.1.3]{golub2013MatrixComputations}, gives, for all $i \in \{1, \dots, p\}$,
\begin{align*}
    \| \partial_k\A_i(x_i) \|^2 &\leq \sum_{m = 1}^q |\big\langle\tfrac{\partial^k}{\partial x_i^k} g(x_i, \bm{t} - t_{i, n}), \tfrac{\partial^k}{\partial x_i^k}g(x_i, \bm{t} - t_{i, m} )\big\rangle| \\
    &\leq \sum_{m \in \Z} \varrho_k(x, |m|\Delta).
\end{align*}
Taking the supremum over $\Omega$ completes the proof.
\end{proof}

The above bound depends solely on the separation and the kernel class, and is independent of the specific realization of the dictionary. Figure~\ref{fig:coherence} illustrates this uniformity for the Gaussian kernel
\begin{equation*}
    g(x,t) = \exp\!\left(-\tfrac{t^2}{x^2}\right),
\end{equation*}
under the \emph{unit-speed reparametrization}, defined by evaluating the function-valued curve $x \mapsto g(x, \cdot)$ along the arc
\begin{equation*}
x \mapsto s(x) = \int_{x_{\min}}^x \| \tfrac{\partial}{\partial \nu} g(\nu, \cdot) \|_{L^2(\R, \diff t)} \diff \nu,
\end{equation*}
yielding the reparametrized kernel
\begin{equation*}
    k(x,t) \coloneq g(s(x), t),
\end{equation*}
so that variations in $x$ induce proportional changes in $g(x,\cdot)$ in $L^2(\R)$.

For each value of $\Delta$, we generate 100 randomized dictionaries with minimal separation $\Delta$ and compute the spectral constants $\sigma_k$ together with the coherence $\mu_k(\Delta)$. In each panel, the solid curve shows the mean value of $\sigma_k$ as a function of $\Delta$, with shaded regions indicating one standard deviation across realizations, while the dotted curve depicts $\mu_k(\Delta)$. For each $k$, the coherence acts as a uniform envelope for $\sigma_k$, and the bound becomes increasingly tight as the separation grows. Indeed, variance across dictionaries collapses as $\Delta$ increases.

Beyond separation, the dominant factor that determines optimization geometry is the rate of decay of $\mu_k(\Delta)$ as $\Delta$ increases. This rate of decay is governed by the decay rate of kernel tails. This mechanism is illustrated in Fig.~\ref{fig:coherence-decay}. We consider the unit-speed reparametrized $u$-Laplace kernel
\begin{equation*}
g_u(x, t) = \exp\left( -\tfrac{|t|^u}{x^u} \right),
\end{equation*}
for $u=2$ (Gaussian kernel) and $u=5$ (faster decay). Panels (a) and (b) show the kernels and their first derivatives with respect to $x$, respectively. Panels (c) and (d) display the empirical spectral constants $\sigma_k$ together with the corresponding coherences $\mu_k(\Delta)$ as functions of $\Delta$. Faster tail decay yields accelerated coherence decay, leading to smaller spectral constants for the same separation. Therefore, a smaller minimum separation is sufficient to achieve a better-conditioned loss geometry.

\begin{figure}[t]
    \centering
    \includegraphics[width=\linewidth]{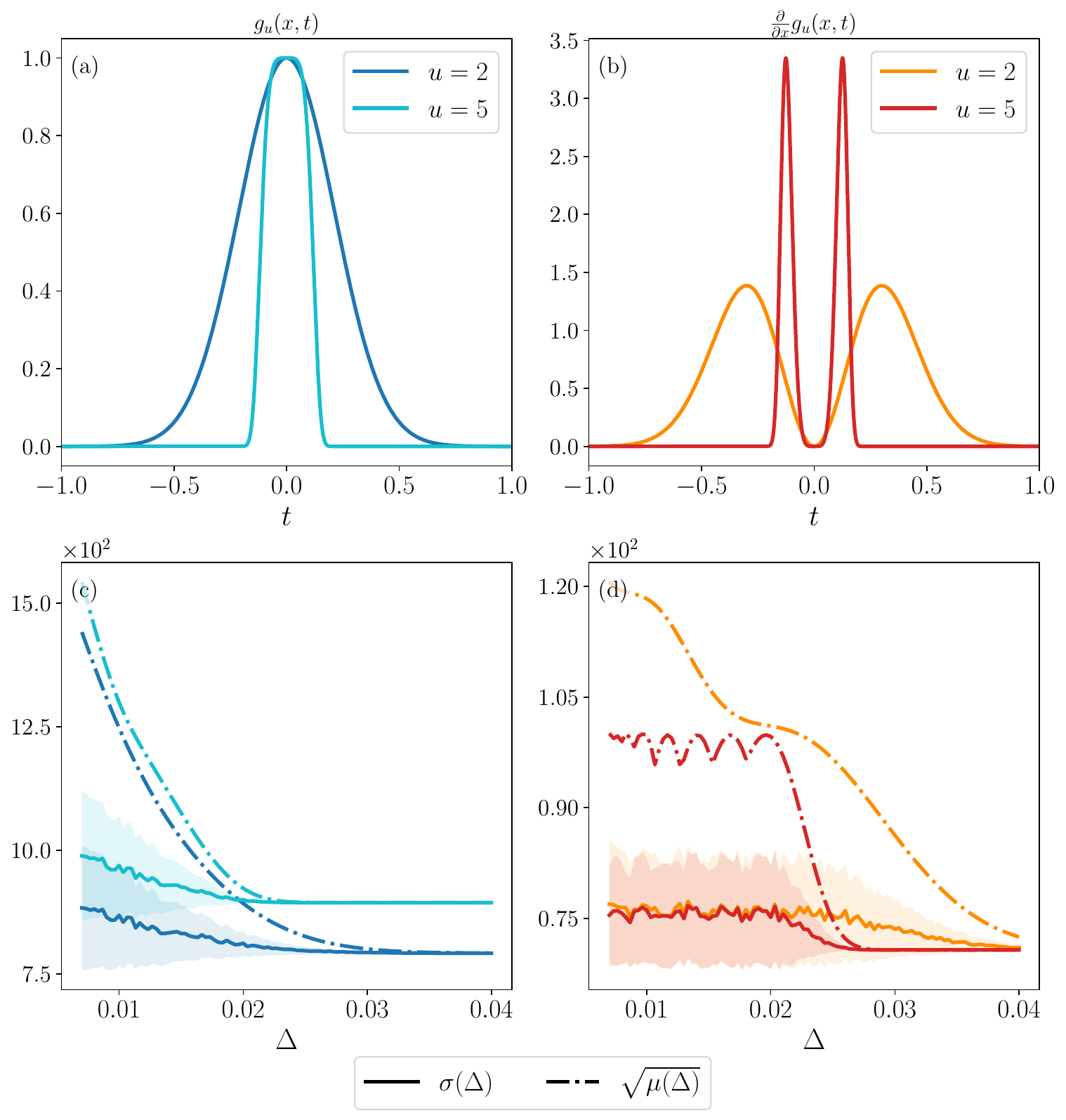}
    \caption{Coherence for kernels with different tail decays. (a) kernels $g_2(x, \cdot)$ and $g_5(x, \cdot)$. (b) first derivatives $\tfrac{\partial}{\partial x} g_2(x, \cdot)$ and $\tfrac{\partial}{\partial x} g_5(x, \cdot)$. (c) Empirical global spectral constant $\sigma_0$ vs coherence $\mu_0$ as a function of $\Delta$. (d) Empirical global spectral constant $\sigma_1$ vs coherence $\mu_1$ as a function of $\Delta$.}
    \label{fig:coherence-decay}
\end{figure}

\section{Numerics}\label{sec:numerics}

We validate the theory with numerical experiments in the PSF unmixing setting. Throughout, we take the feasible set $\Omega = [x_{\min}, x_{\max}]^p$ and the ground truth parameters 
\begin{equation*}
    \bx^\star = \tfrac{1}{2}(x_{\min} + x_{\max})\bm{1}_p, 
    \qquad 
    \by^\star = \bm{1}_{pq}.
\end{equation*}
Random dictionaries with prescribed minimal separation are generated by rejection sampling. Given $pq$ spikes, the first is drawn uniformly on $I=[-T/2,T/2]$, the second is placed at distance $\pm \Delta$, and each subsequent spike is sampled uniformly while rejecting locations within $\Delta$ of existing spikes. We use $N=10^4$ and $\Delta=5\cdot10^{-3}$ throughout.

We study the evolution of the local geometry of problems~\eqref{eq:lstq} and~\eqref{eq:projected-lstq} with the number of nonlinear parameters $p$ and block size $q$. For each configuration $(p,q)$, a dictionary is sampled once and kept fixed.

Code for reproducibility is available at~\url{https://github.com/smichelena/PsfUnmixing}.

\subsection{Empirical Basin Size Estimates}

The strong basin size of the PSF unmixing problem described in Section~\ref{sec:psf-unmixing} cannot be computed exactly. We therefore probe the landscape via Monte Carlo sampling and compare empirical estimates with the analytical predictions.

For each experiment, 100 perturbed parameters are drawn. Perturbation magnitudes are measured in the unmixing metric~\eqref{eq:rho} for least squares and in the $\ell_2$ metric on $\Omega$ for the projected estimator. For each sample, we compute the minimum Hessian eigenvalue and the corresponding Weyl-envelope bounds~\eqref{eq:weyl}, and compare them with the theoretical bounds from Theorem~\ref{thm:least-squares-basin} in the least squares case and~\ref{thm:projected-envelope} for variable projection. The empirical strong convexity radius is the largest perturbation for which the Hessian remains positive definite over all samples.
Figures~\ref{fig:ls-envelopes} and~\ref{fig:vp-envelopes} report results for noiseless least squares and variable projection. The blue curve shows the minimum eigenvalue, the orange dashed curve the Weyl-envelope estimate, and the red dotted curve the analytical bound. For variable projection, the green dash–dot curve shows the restricted envelope bound~\eqref{eq:projected-envelope}.

\begin{figure}[t]
    \centering
    \begin{tikzpicture}
    
    \node (env) {\includegraphics[width=\linewidth]{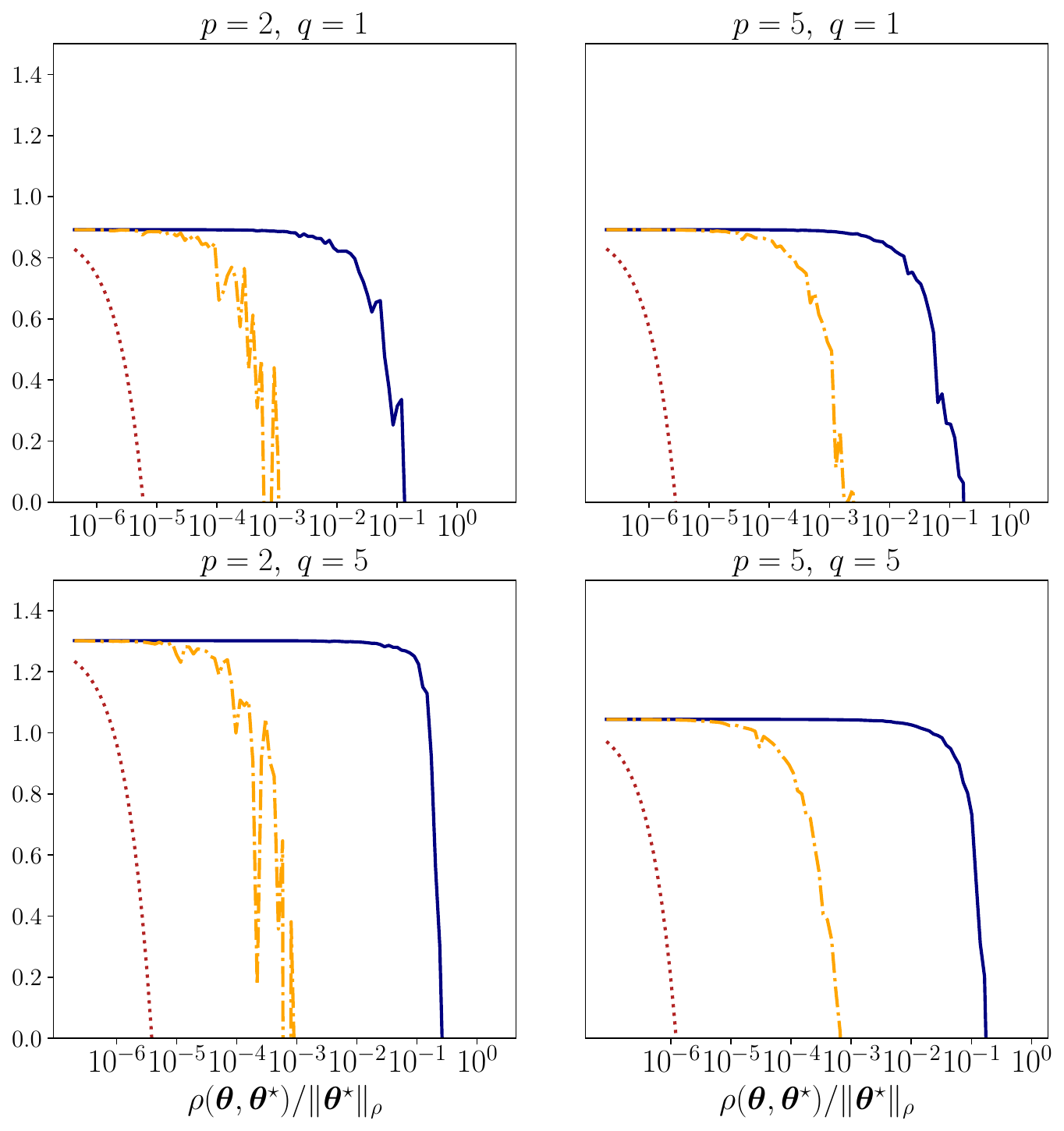}};
    
    \end{tikzpicture}
    \caption{Noiseless empirical basin estimates for least squares. Blue: sampled minimum eigenvalue. Orange dashed: Weyl envelope estimate. Red dotted: analytical lower bound.}
    \label{fig:ls-envelopes}
\end{figure}

\begin{figure}[t]
    \centering
    \begin{tikzpicture}
    
    \node (env) {\includegraphics[width=\linewidth]{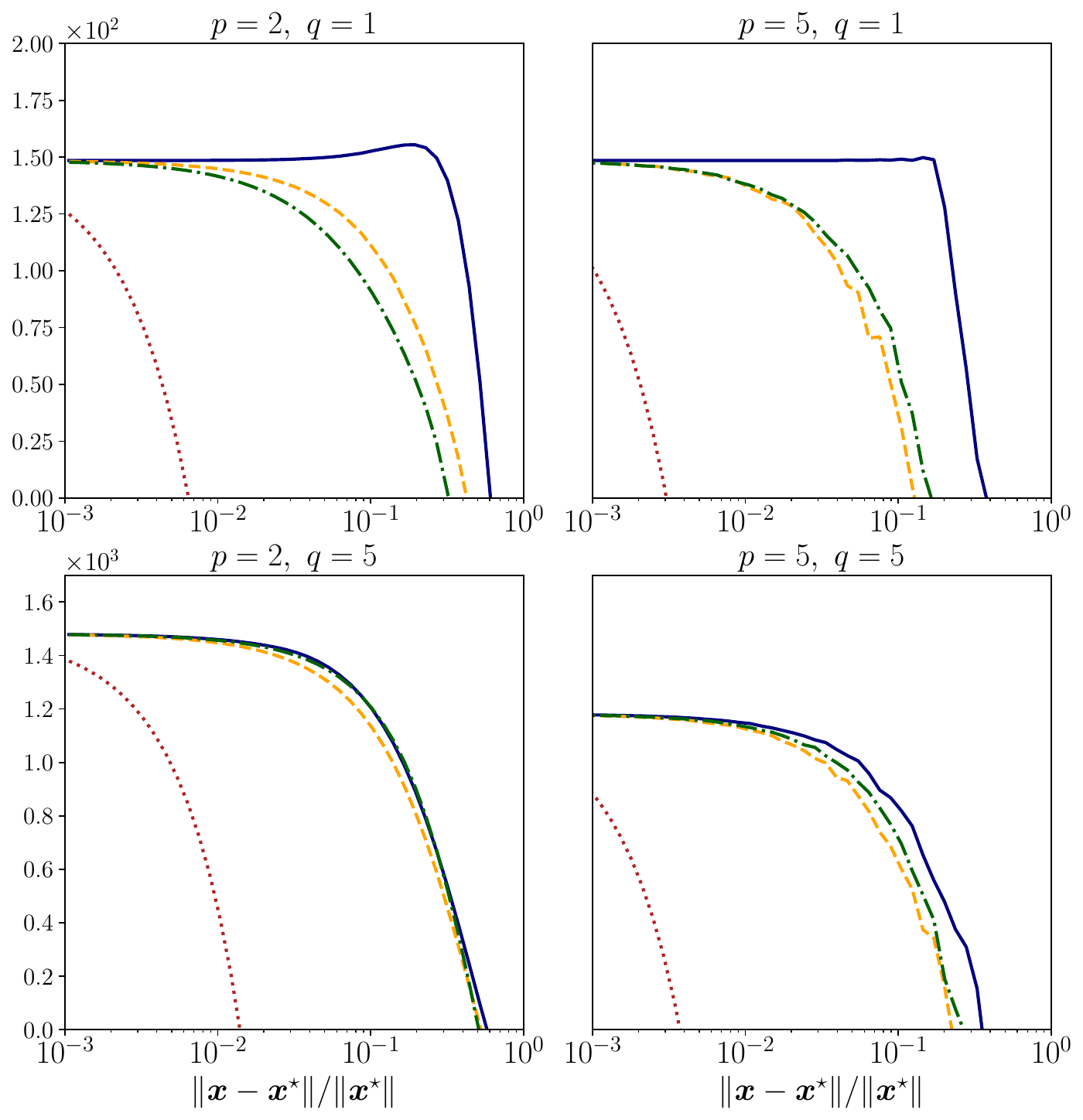}};
    
    \end{tikzpicture}
    \caption{Noiseless empirical basin estimates for projected least squares. Blue: the minimum eigenvalue. Orange dashed: the Weyl-envelope estimate. Green dashed: the Weyl-envelope bound using the full Hessian restricted to the optimal value manifold~\ref{thm:projected-envelope}. Red dotted: and the analytical prediction.}
    \label{fig:vp-envelopes}
\end{figure}

We repeat the experiments for $(p,q)=(2,1)$ and $(p,q)=(5,5)$ with additive Gaussian noise at $\mathrm{SNR}=0\,\mathrm{dB}$ over 30 realizations, using the same dictionaries as in the noiseless case.
Figures~\ref{fig:noisy-envelopes} and~\ref{fig:varpro-noisy-envelopes} summarize the results for least squares and variable projection. Black curves show noiseless baselines. Colored curves denote means across realizations, and shaded regions indicate one standard deviation.

\begin{figure}[t]
    \centering
    \begin{tikzpicture}
    
    \node (env) {\includegraphics[width=\linewidth]{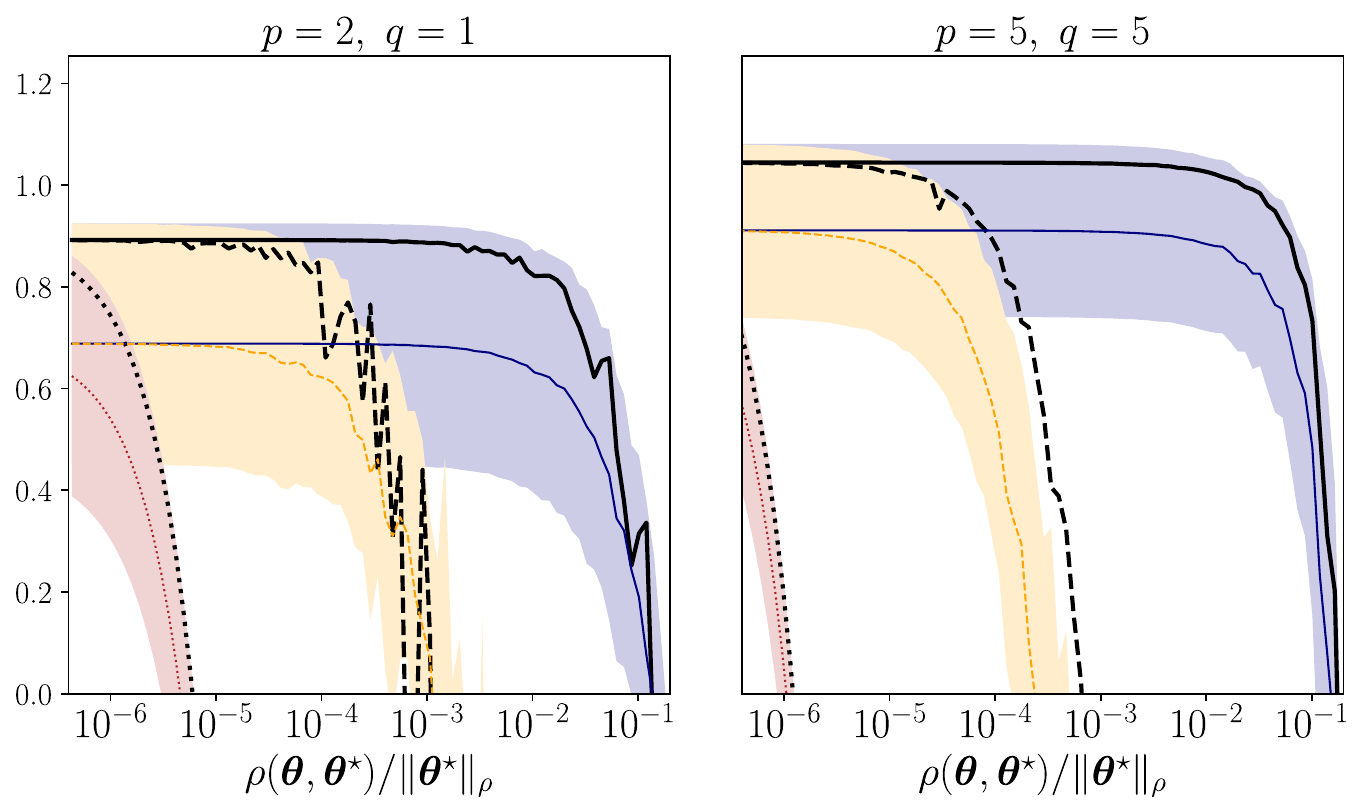}};
    
    \end{tikzpicture}
    \caption{Noisy empirical basin estimates for least squares at 0 $\mathrm{dB}$. Black curves show the noiseless baseline. Colored curves indicate mean values across noise realizations. Blue: true minimum eigenvalue. Orange dashed: Weyl-envelope estimate. Red dotted: and analytical prediction. Shaded regions represent one standard deviation over noise realizations.}
    \label{fig:noisy-envelopes}
\end{figure}

\begin{figure}[t]
    \centering
    \begin{tikzpicture}
    
    \node (env) {\includegraphics[width=\linewidth]{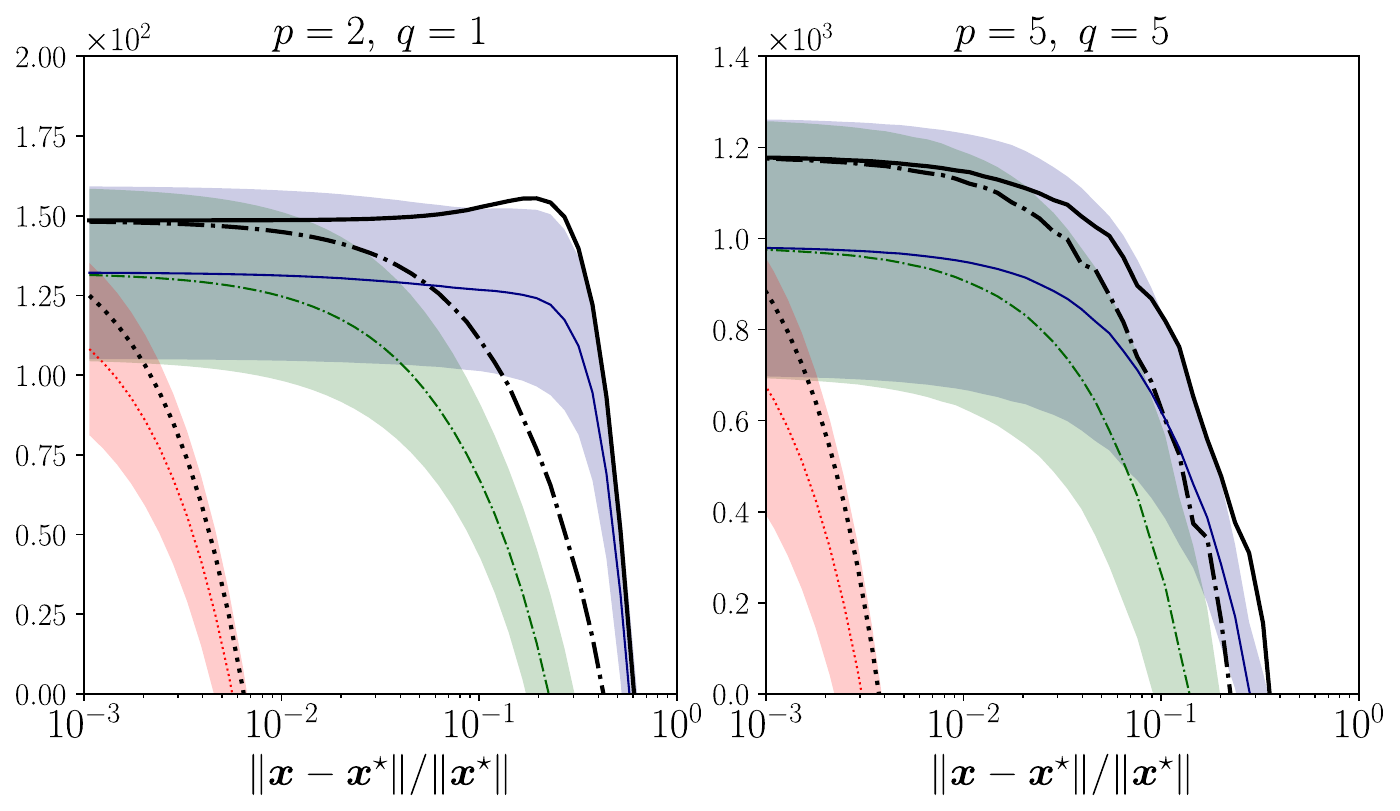}};
    
    \end{tikzpicture}
    \caption{Noisy empirical basin estimates for projected least squares at 0 $\mathrm{dB}$. Black curves show the noiseless baselines. Colored curves indicate mean values across noise realizations. Blue: the minimum eigenvalue. Orange dashed: the Weyl-envelope estimate. Green dashed: the Weyl-envelope bound using the full Hessian restricted to the optimal value manifold. Red dotted: and the analytical prediction. Shaded regions represent one standard deviation over noise realizations.}
    \label{fig:varpro-noisy-envelopes}
\end{figure}

As $q$ increases, the minimum eigenvalue of the Hessian matrix grows, reflecting improved conditioning from larger column blocks and hence more information per nonlinear parameter. For variable projection, the Weyl envelope and restricted envelope become tight on the minimum eigenvalue, indicating sharper bounds and better overall conditioning. The projected minimum eigenvalue is also visibly amplified. Furthermore, we observe that this behavior is robust to strong noise. These experiments serve as strong validation of section~\ref{sec:local-geometry}.

\subsection{Stability}

We evaluate recovery stability of the unprojected and projected least square estimators defined in p
Problems~\eqref{eq:lstq} and~\eqref{eq:projected-lstq} for $(p,q)=(2,1)$ and $(p,q)=(5,5)$. Measurements are corrupted with Gaussian noise for $\mathrm{SNR}\in[-10,20]\,\mathrm{dB}$, with 100 realizations per noise level. Both estimators are initialized inside their basins and recovered using the Levenberg--Marquardt algorithm. Empirical recovery errors are measured in the unmixing metric $\rho$ given by Equation~\eqref{eq:rho}. 
Analytical stability bounds predicted by Corollary~\ref{cor:ls-stability} and Lemma~\ref{cor:vp-stability} are evaluated across the same SNR range and compared with empirical errors.

Figure~\ref{fig:stability} shows that, when initialized in the basin, both methods converge to the same solution and achieve identical empirical performance. However, the upper bounds on the recovery error of least squares, predicted by Corollary~\ref{cor:ls-stability} are highly pessimistic, overestimating errors by up to four orders of magnitude. In contrast, the bounds on the projected least squares recovery error, predicted by Lemma~\ref{cor:vp-stability} closely track empirical behavior.

This gap stems from Jacobian conditioning. For least squares, the Jacobian condition numbers at the ground truth are approximately $27$ for $(p,q) =(2,1)$ and $45$ for $(p,q) = (5,5)$, leading to substantial amplification in the stability bound through factors proportional to $\sigma_{\min}(\bm{J}(\btheta^\star))^{-4}$. Under variable projection, these collapse to about $1.4$ and $1.7$, explaining the significant improvement in predictive accuracy of the projected bounds. These results show that variable projection fundamentally improves the conditioning governing recovery stability, not merely optimization behavior.

\begin{figure}[t]
    \centering
    \includegraphics[width=\linewidth]{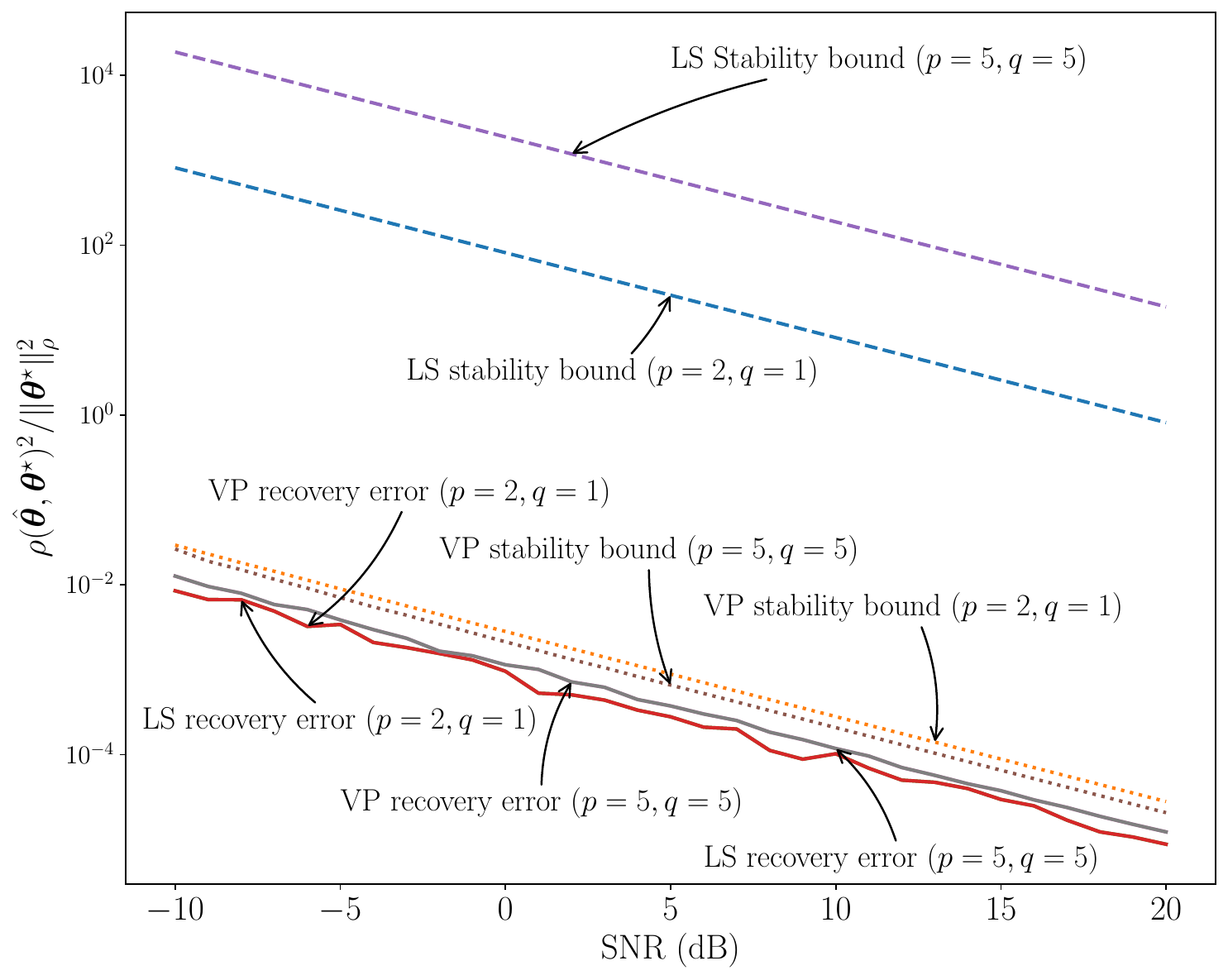}
    \caption{Empirical recovery error and analytical stability bounds in the PSF unmixing setting for $(p,q)=(2,1)$ and $(p,q)=(5,5)$ across SNR levels ranging from $-10,\mathrm{dB}$ to $20,\mathrm{dB}$. Solid curves show empirical mean recovery error measured in the unmixing metric $\rho$~\eqref{eq:rho} over 100 noise realizations, while dashed curves show the analytical stability bounds predicted by theory.}
    \label{fig:stability}
\end{figure}

\subsection{Algorithmic Convergence and Coherence}

In this section, we demonstrate that the local geometry predicted by the theory accurately characterizes algorithmic behavior. Specifically, we show that coherence governs optimization through its effect on convexity: the strong convexity radius acts as a reliable predictor of the empirical convergence region, and its alignment with the latter determines the efficiency of local optimization. We validate this mechanism by comparing analytical predictions with empirical behavior, both in terms of convergence regions and convergence rates.

We compare the analytical upper and lower bounds on the strong convexity radius from Theorem~\ref{thm:radii-comparison} with the empirical convergence radius of Levenberg--Marquardt at $\mathrm{SNR}=10\,\mathrm{dB}$, estimated via Monte Carlo initializations. For each perturbation distance, the algorithm is initialized 100 times, and the minimum Hessian eigenvalue is computed in parallel to estimate the empirical strong convexity radius.

Recovery is declared successful if the reconstructed parameter lies within the stability tolerance predicted by Corollary~\ref{cor:vp-stability}, measured in the metric $\rho$. The empirical convergence radius is defined as the largest perturbation for which the success rate equals one.

Experiments are conducted for $(p,q)=(2,1)$ and $(5,5)$, and for $u\in \{0.5,1,2,5\}$. The results, reported in Fig.~\ref{fig:convergence_regions}, show that the predicted convexity radii accurately characterize the empirical convergence region. In particular, coherence governs this behavior: smaller $u$ (higher coherence) induces a separation between the convexity and convergence radii, while lower coherence leads to close alignment and improved optimization performance. The analytical interval remains tight across all configurations and is tightest for $(p,q)=(2,1)$ and $u=0.5$.

\subsection{Empirical Convergence of Nonconvex Solvers}

We now examine how the local geometry of problems~\eqref{eq:lstq} and~\eqref{eq:projected-lstq}, induced by coherence, translates into algorithmic dynamics. We compare the empirical convergence behavior of several nonconvex solvers—Levenberg--Marquardt (blue triangles), Gauss--Newton (orange circles), and gradient descent (green squares)—by tracking the gradient norm as a function of iteration count for both least squares and variable projection.

Experiments are conducted for $(p,q)=(2,1)$ and $(5,5)$, and for $u\in \{0.5,5\}$, as shown in Fig.~\ref{fig:optim-behavior}. Faster kernel tail decay ($u=5$), which induces lower coherence, leads to significantly accelerated convergence across all methods, consistent with the predicted improvement in conditioning. In contrast, the regime $(p,q) =(2,1)$ with $u=0.5$, where convexity and convergence radii are misaligned, exhibits markedly slower convergence, highlighting the impact of coherence on both the size of the convergence region and the efficiency of local optimization.

\begin{figure}[t]
    \centering
    \includegraphics[width=\linewidth]{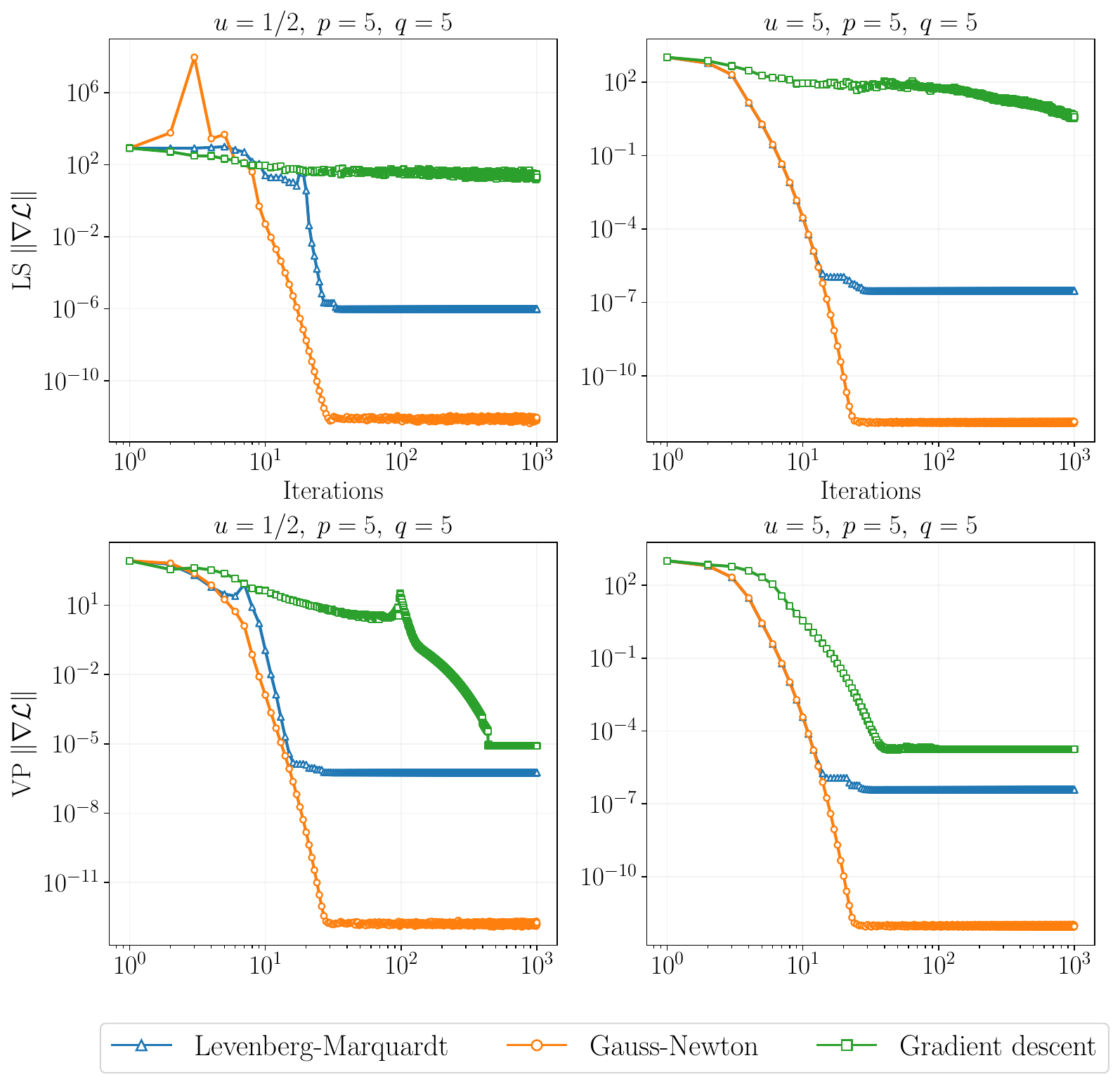}
    \caption{Empirical convergence behavior of Levenberg--Marquardt (blue triangles), Gauss--Newton (orange circles), and gradient descent (green squares) for least squares and variable projection for $(p,q)=(5,5)$ and for $u=0.5$ and $u=5$.}
    \label{fig:optim-behavior}
\end{figure}

\section{Conclusion, Perspectives}\label{sec:conclusion}

Separable signal models encompass a wide range of applications in signal processing. We have shown, using classical perturbation results, that the size of the strong basin of attraction is characterized by the Lipschitz constants of the forward map and its Fréchet derivatives. In the PSF unmixing case, we have shown that these constants depend on separation and that coherence directly determines the optimization geometry. The found constants are written as global operator-norm bounds over the feasible set for the nonlinear parameters; this, by nature, yields comparatively wide bounds in many applications. As such, future research may focus on refining the bounds by both using better specialized perturbation bounds and by, instead of expressing the region of favorable geometry in terms of classical strong convexity, harnessing generalized, descent direction-adapted notions of convexity. Finally, sparsity-informed formulations of manifold unmixing may be explored, as done for example in~\cite{michelena_gretsi_2025}. Necessary and sufficient conditions for convergence, in terms of interpretable quantities, constitute an attractive result.

\begin{figure*}[t]
    \centering
    \includegraphics[width=\textwidth]{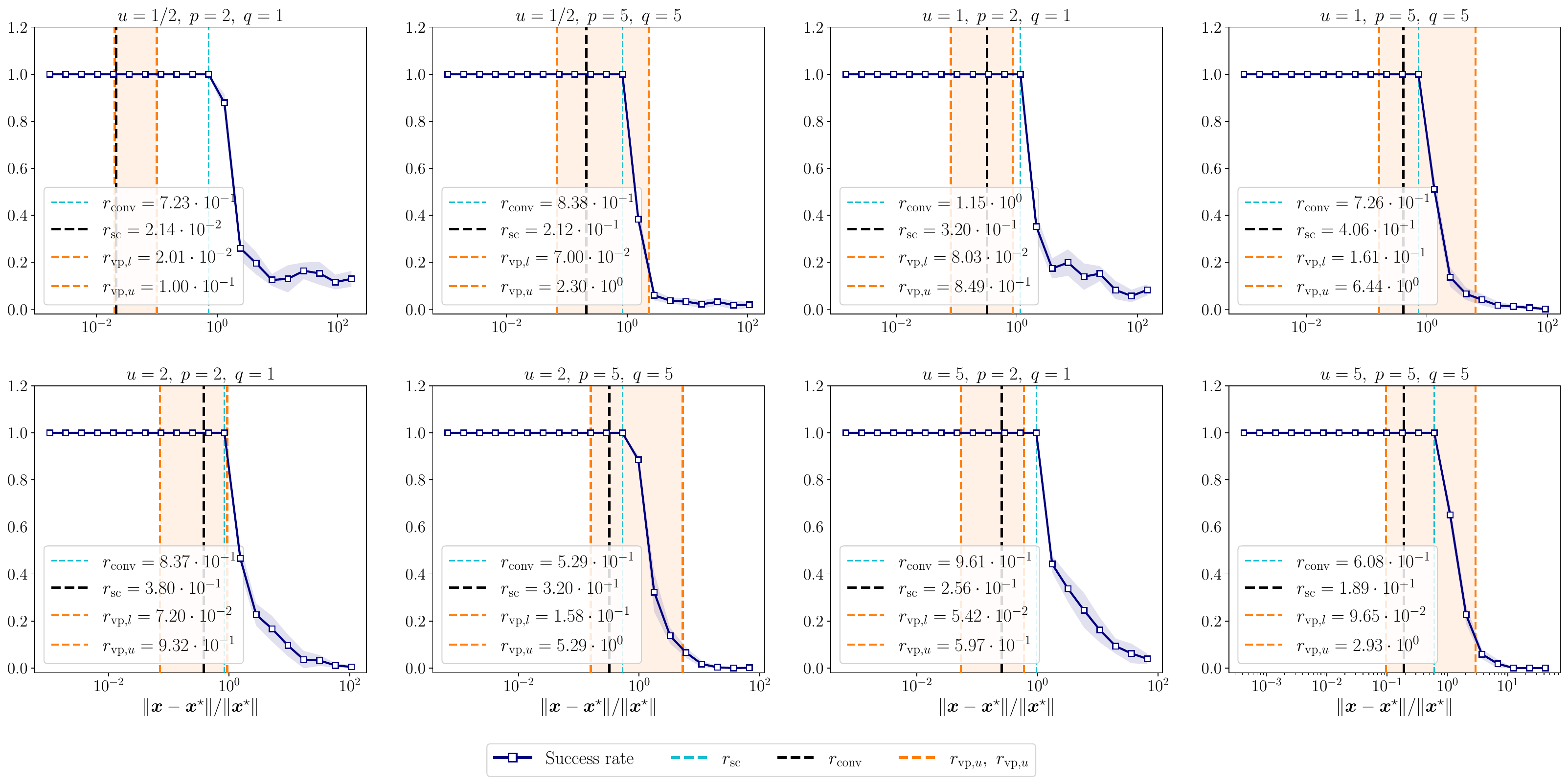}
    \caption{Empirical convergence radius and analytical strong convexity bounds for variable projection at $\mathrm{SNR}=10\,\mathrm{dB}$. Blue squares show the empirical recovery success rate, with shaded regions indicating one standard deviation. The dashed black vertical line marks the empirical strong convexity radius, and the dashed cyan vertical line the empirical convergence radius. Dashed orange vertical lines and shaded region in-between denote the lower and upper analytical convergence radius bounds from Theorem~\ref{thm:radii-comparison}. Results are shown for $(p,q)=(2,1)$ and $(5,5)$ and for $u\in\{0.5,1,2,5\}$.}
    \label{fig:convergence_regions}
\end{figure*}

\appendices

\section{Perturbation Bounds}

\subsection{Preliminaries}

Denote the Fréchet derivative of the residual $\Dd \br$. For any pair $(\bx, \by) = \btheta$, the derivative $\Dd \br(\bx, \by)$ is a linear map from $\Omega \times \R^{d}$ to $\R^N$. For a direction in $\bx$ $\bu_1$ and a direction in $\by$ $\bu_1$ it is given by
\begin{equation*}
    \Dd \br(\btheta) [\bu_1, \bu_2] = (\Dd \A(\bx)[\bu_1])\by + \A(\bx)\bu_2. 
\end{equation*}

Similarly, for any $\btheta$, the second Derivative $\Dd^2 \br(\alpha)$ is a bilinear form from $(\Omega \times \R^{d}) \times (\Omega \times \R^{d})$ to $\R^N$. For total directions $\bu = (\bu_1, \bu_2)$ and $\bv = (\bv_1, \bv_2)$ it is given by
\begin{multline}
    \Dd^2 \br(\btheta)[(\bu_1,\bu_2),(\bv_1,\bv_2)]
    = (\Dd^2 \A(\bx)[\bu_1,\bv_1])\by \\
     + \Dd \A(\bx)[\bu_1]\bv_2
      + \Dd \A(\bx)[\bv_1]\bu_2.
\end{multline}

Finally, the fundamental theorem of calculus gives
\begin{equation*}
    \Dd^k \A(\bx^\star) - \Dd^k \A(\bx) = \int_0^1 \Dd^{k+1} \A(\bx + t\Delta \bx)[\Delta \bx] \diff t,
\end{equation*}
with $\Delta \bx \coloneq \bx^\star - \bx$. So that, for all $k \in \{1, \dots, 3\}$ we obtain the bounds
\begin{align*}
    \|\Dd^k \A(\bx^\star)) - \Dd^k \A(\bx)\| &\leq \sup_{t \in [0,1]}\|\Dd^{k+1}\A(\bx + t\Delta \bx)\| \|\Delta \bx\| \\
    &\leq \sup_{\bx \in \Omega}\|\Dd^{k+1}\A(\bx)\| \|\Delta \bx\|,
\end{align*}
by convexity of $\Omega$.

\subsection{Jacobian Perturbation}\label{app:jacobian-perturbation}

For any $\btheta  \in \Omega \times \R^{d}$, the Jacobian is given by $\Dd \br$ represented as a matrix in the canonical basis. Therefore, for any total direction $\bu = (\bu_1, \bu_2)$ we have 
\begin{align*}
\|\bm{J}(\btheta) - \bm{J}(\btheta^\star)\| &= \sup_{\|\bu\|=1} \| (\Dd\br(\btheta) - \Dd \br(\btheta^\star))[\bu]\| \\
&\hspace{-20pt}\leq \sup_{\|\bu\|=1} \| (\Dd \A(\bx)[\bu_1])\by - (\Dd \A(\bx^\star)[\bu_1])\by^\star \| \\
&+ \sup_{\|\bu\|=1} \| (\A(\bx) - \A(\bx^\star))\bu_2 \|.
\end{align*}
Moreover,
\begin{align*}
(\Dd \A(\bx)[\bu_1])\by - (\Dd \A(\bx^\star)&[\bu_1])\by^\star
= (\Dd \A(\bx)[\bu_1])(\by-\by^\star) \\
&+ ((\Dd \A(\bx) - \Dd \A(\bx^\star))[\bu_1])\by^\star.
\end{align*}
then 
\begin{align*}
    \|\bm{J}(\btheta) - \bm{J}(\btheta^\star)\| \leq (\sigma_2 \|\by^\star\| + \sigma_1)\|\bx - \bx^\star\| + \sigma_1\|\by - \by^\star\|.
\end{align*}

\subsection{Residual Hessian Perturbation}\label{app:residual-perturbation}

For any total direction $\bu = (\bu_1, \bu_2)$, the perturbation of the residual part of the Hessian satisfies, 
\begin{multline*}
    \hspace{-6pt}\langle (\H_r(\btheta^\star) - \H_r(\btheta))\bu, \bu\rangle 
    = \langle \Dd^2\br(\btheta^\star)[\bu, \bu], \br(\btheta^\star) - \br(\btheta)\rangle \\ + \langle (\Dd^2\br(\btheta^\star) - \Dd^2\br(\btheta))[\bu, \bu],\br(\btheta)\rangle.
\end{multline*}
The first term satisfies
\begin{align*}
    &\langle \Dd^2\br(\btheta^\star)[\bu, \bu], \br(\btheta^\star) - \br(\btheta)\rangle \\
    &\leq (\|\Dd^2\A(\bx^\star)\|\|\by^\star\| + 2\|\Dd \A(\bx^\star)\|) \\
    &\hspace{10pt}\cdot (\sup_{\bx \in \Omega}\|\A(\bx)\|\|\by - \by^\star\| + \sup_{\bx \in \Omega}\|\Dd\A(\bx)\|\|\by^\star\|\|\bx - \bx^\star\|)
\end{align*}
For any $\btheta \in \Omega \times \R^{d}$ the residual satisfies
\begin{multline*}
    \|\br(\btheta)\| \leq \|\bw\| + \sup_{\bx \in \Omega}\|\A(\bx)\|\|\by - \by^\star\| 
    \\+ \sup_{\bx \in \Omega}\|\Dd\A(\bx)\|\|\by^\star\|\|\bx - \bx^\star\|.
\end{multline*}
Finally, we have the bound
\begin{align*}
    \MoveEqLeft[0]\|(\Dd^2\br(\btheta^\star) - \Dd^2\br(\btheta))[\bu, \bu]\| & \\
    &\qquad \leq \sup_{\bx \in \Omega}\|\Dd^3\A(\bx)\|\|\by^\star\|\|\bx - \bx^\star\| \\
    & \qquad + \sup_{\bx \in \Omega}\|\Dd^2\A(\bx)\|(2\|\bx - \bx^\star\| +\|\by - \by^\star\|).
\end{align*}
Combining all the previous bounds yields the result.

\section{Radius of the Basin of Variable Projection }\label{app:vp-radius}

Combine lemma~\ref{thm:projected-envelope} with inequality~\eqref{eq:rho-lift-bound} to yield a quadratic, decreasing, and concave polynomial on the distance $\eps =  \|\bx - \bx^\star\|$ that reads
This yields a quadratic lower bound of the form
\begin{equation*}
    q(\eps) = \lambda(\btheta^\star, \bw)  - Kc_{1, \mathrm{vp}}(\btheta^\star, \bw) \eps - Kc_{2, \mathrm{vp}}(\btheta^\star, \bw) \eps^2,
\end{equation*}
where
\begin{align*}
    c_{1, \mathrm{vp}}(\btheta^\star, \bw) &\coloneq c_\mathrm{vp}(c_1(\btheta^\star) + 2c_2(\btheta^\star))(1 + \sigma_{-}^{-1})\|\bw\|, \\
    c_{2, \mathrm{vp}}(\btheta^\star, \bw) &\coloneq c_2(\btheta^\star) c_\mathrm{vp}^2, \\
    \lambda(\btheta^\star, \bw) &\coloneq \lambda_{\min}(\H_\mathrm{vp}(\bx^\star)) - c_1(\btheta^\star)(1 + \sigma_{-}^{-1})\|\bw\| \\
    &\hspace{60pt} - (1 + \sigma_{-}^{-1})^2\|\bw\|^2.
\end{align*}
The first positive root of $q$ gives the analytical radius of strong convexity for the variable projection estimator
\begin{equation*}
    r_\mathrm{vp} = \tfrac{Kc_{1, \mathrm{vp}}(\btheta^\star, \bw) - \sqrt{K^2c_{1, \mathrm{vp}}(\btheta^\star, \bw)^2 + 4Kc_{2, \mathrm{vp}}(\btheta^\star, \bw)\lambda(\btheta^\star, \bw)}}{2Kc_{2, \mathrm{vp}}(\btheta^\star, \bw)}.
\end{equation*}
In the noiseless case ($\bw = 0$) and $K \approx 1$ limit the radius reads
\begin{equation*}
    \eps_\mathrm{vp} = \tfrac{\sqrt{c_{1}(\btheta^\star)^2 + 4c_{2}(\btheta^\star)k_\mathrm{vp}\lambda_{\min}(\H(\btheta^\star))} - c_{1}(\btheta^\star)}{2c_{2}(\btheta^\star)c_\mathrm{vp}}.
\end{equation*}
We compare the numerators of $\eps_\mathrm{vp}$ and $\eps_\mathrm{LS}$, giving (we omit arguments for notational simplicity)
\begin{align*}
    \frac{2c_2r_\mathrm{vp}c_\mathrm{cv}}{2r_\mathrm{ls}c_2} &= \tfrac{\sqrt{c_{1}^2 + 4c_{2}k_\mathrm{vp}\lambda_{\min}(\H)} - c_{1}}{\sqrt{c_{1}^2 + 4c_{2}\lambda_{\min}(\H)} - c_{1}} \\
    &= k_\mathrm{vp} \tfrac{\sqrt{c_{1}^2 + 4c_{2}\lambda_{\min}(\H)} + c_{1}}{\sqrt{c_{1}^2 + 4c_{2}k_\mathrm{vp}\lambda_{\min}(\H)} + c_{1}} \leq k_\mathrm{vp}.
\end{align*}
Furthermore, 
\begin{equation*}
   \frac{\sqrt{c_{1}^2 + 4c_{2}\lambda_{\min}(\H)} + c_{1}}{\sqrt{c_{1}^2 + 4c_{2}k_\mathrm{vp}\lambda_{\min}(\H)} + c_{1}} \geq \frac{1}{\sqrt{k_\mathrm{vp}}},
\end{equation*}
so that the numerators satisfy
\begin{equation*}
    \sqrt{k_\mathrm{vp}}r_\mathrm{ls} c_2 \eps_\mathrm{vp} \leq c_2r_\mathrm{vp}c_\mathrm{vp} \leq k_\mathrm{vp} r_\mathrm{ls} c_2.
\end{equation*}
Finally, the above inequality produces a radius in the topology induced by $(\bx, \bx') \mapsto \rho(\btheta(\bx), \btheta(\bx'))$. To yield a radius in the Euclidean topology, one must divide by $c_\mathrm{vp}$. Yielding the result.

\IEEEtriggeratref{4}
\renewcommand*{\bibfont}{\footnotesize}
\printbibliography

@article{SILVA2016178,
  title = {Kantorovich’s Theorem on Newton’s Method for Solving Generalized Equations under the Majorant Condition},
  author = {Silva, Gilson N.},
  date = {2016},
  journaltitle = {Applied Mathematics and Computation},
  shortjournal = {Appl. Math. Comput.},
  volume = {286},
  pages = {178--188},
  issn = {0096-3003},
  doi = {10.1016/j.amc.2016.04.015},
  url = {https://www.sciencedirect.com/science/article/pii/S009630031630265X},
  langid = {english}
}

@article{poon2023GeometryOfftheGrida,
  title = {The {{Geometry}} of {{Off-the-Grid Compressed Sensing}}},
  author = {Poon, Clarice and Keriven, Nicolas and Peyré, Gabriel},
  date = {2023-02-01},
  journaltitle = {Foundations of Computational Mathematics},
  shortjournal = {Found. Comput. Math.},
  volume = {23},
  number = {1},
  pages = {241--327},
  issn = {1615-3383},
  doi = {10.1007/s10208-021-09545-5},
  url = {https://doi.org/10.1007/s10208-021-09545-5},
  langid = {english}
}

@article{li2016identifiability,
  title={Identifiability in blind deconvolution with subspace or sparsity constraints},
  author={Li, Yanjun and Lee, Kiryung and Bresler, Yoram},
  journal={IEEE Transactions on information Theory},
  volume={62},
  number={7},
  pages={4266--4275},
  year={2016},
  publisher={IEEE}
}

@article{li2019multichannel,
  title={Multichannel sparse blind deconvolution on the sphere},
  author={Li, Yanjun and Bresler, Yoram},
  journal={IEEE Transactions on Information Theory},
  volume={65},
  number={11},
  pages={7415--7436},
  year={2019},
  publisher={IEEE}
}

@article{hockmann2023weak,
  title={Weak Sparse Superresolution is Well-Conditioned},
  author={Hockmann, Mathias and Kunis, Stefan},
  journal={SIAM Journal on Imaging Sciences},
  volume={16},
  number={1},
  pages={SC1--SC13},
  year={2023},
  publisher={SIAM}
}

@article{ferreira2020stable,
  title={On the stable resolution limit of total variation regularization for spike deconvolution},
  author={Ferreira Da Costa, Maxime  and Chi, Yuejie},
  journal={IEEE Transactions on Information Theory},
  volume={66},
  number={11},
  pages={7237--7252},
  year={2020},
  publisher={IEEE}
}

@article{traonmilin2024strong,
  title={On strong basins of attractions for non-convex sparse spike estimation: upper and lower bounds},
  author={Traonmilin, Yann and Aujol, Jean-Fran{\c{c}}ois and B{\'e}nard, Pierre-Jean and Leclaire, Arthur},
  journal={Journal of Mathematical Imaging and Vision},
  volume={66},
  number={1},
  pages={57--74},
  year={2024},
  publisher={Springer}
}

@article{traonmilin2020basins,
  title={The basins of attraction of the global minimizers of the non-convex sparse spike estimation problem},
  author={Traonmilin, Yann and Aujol, Jean-Fran{\c{c}}ois},
  journal={Inverse Problems},
  volume={36},
  number={4},
  pages={045003},
  year={2020},
  publisher={IOP Publishing}
}

@article{chi2020harnessing,
  title={Harnessing sparsity over the continuum: Atomic norm minimization for superresolution},
  author={Chi, Yuejie and Ferreira Da Costa, Maxime},
  journal={IEEE Signal Processing Magazine},
  volume={37},
  number={2},
  pages={39--57},
  year={2020},
  publisher={IEEE}
}

@article{knudson2014inferring,
  title={Inferring sparse representations of continuous signals with continuous orthogonal matching pursuit},
  author={Knudson, Karin C and Yates, Jacob and Huk, Alexander and Pillow, Jonathan W},
  journal={Advances in neural information processing systems},
  volume={27},
  year={2014}
}

@article{huang2008three,
  title={Three-dimensional super-resolution imaging by stochastic optical reconstruction microscopy},
  author={Huang, Bo and Wang, Wenqin and Bates, Mark and Zhuang, Xiaowei},
  journal={Science},
  volume={319},
  number={5864},
  pages={810--813},
  year={2008},
  publisher={American Association for the Advancement of Science}
}

@article{variableProjection,
	ISSN = {00361429},
	URL = {http://www.jstor.org/stable/2156365},
	author = {G. H. Golub and V. Pereyra},
	journal = {SIAM Journal on Numerical Analysis},
	number = {2},
	pages = {413--432},
	publisher = {Society for Industrial and Applied Mathematics},
	title = {The Differentiation of Pseudo-Inverses and Nonlinear Least Squares Problems Whose Variables Separate},
	urldate = {2024-10-21},
	volume = {10},
	year = {1973}
}

@book{Nocedal2006,
  author    = {Jorge Nocedal and Stephen J. Wright},
  title     = {Numerical Optimization},
  edition   = {2nd},
  year      = {2006},
  publisher = {Springer},
  address   = {New York, NY, USA},
  isbn      = {978-0-387-30303-1}
}

@inproceedings{moulines_maximum_1997,
	title = {Maximum likelihood for blind separation and deconvolution of noisy signals using mixture models},
	volume = {5},
	url = {https://ieeexplore.ieee.org/abstract/document/604649},
	doi = {10.1109/ICASSP.1997.604649},
	eventtitle = {1997 {IEEE} International Conference on Acoustics, Speech, and Signal Processing},
	pages = {3617--3620 vol.5},
	booktitle = {1997 {IEEE} International Conference on Acoustics, Speech, and Signal Processing},
	author = {Moulines, E. and Cardoso, J.-F. and Gassiat, E.},
	urldate = {2025-02-01},
	date = {1997-04},
	note = {{ISSN}: 1520-6149},
}

@ARTICLE{costa_local_2023,
  author={Ferreira Da Costa, Maxime and Chi, Yuejie},
  journal={IEEE Journal on Selected Areas in Information Theory}, 
  title={Local Geometry of Nonconvex Spike Deconvolution From Low-Pass Measurements}, 
  year={2023},
  volume={4},
  number={},
  pages={1-15},
  doi={10.1109/JSAIT.2023.3262689}}

@book{golub2013MatrixComputations,
  title = {Matrix Computations},
  author = {Golub, Gene H. and Van Loan, Charles F.},
  date = {2013},
  series = {Johns {{Hopkins}} Studies in the Mathematical Sciences},
  edition = {Fourth edition},
  publisher = {The Johns Hopkins University Press},
  location = {Baltimore},
  isbn = {978-1-4214-0794-4},
  langid = {english},
  pagetotal = {756}
}

@article{Chretien2020,
  author    = {S{\'e}bastien Chr{\'e}tien and Himanshu Tyagi},
  title     = {Multi-kernel Unmixing and Super-Resolution Using the Modified Matrix Pencil Method},
  journal   = {Journal of Fourier Analysis and Applications},
  volume    = {26},
  number    = {1},
  pages     = {18},
  year      = {2020},
  doi       = {10.1007/s00041-020-09725-x},
  url       = {https://doi.org/10.1007/s00041-020-09725-x}
}

@article{projected-hessian,
  author    = {Axel Ruhe and Per {\AA}ke Wedin},
  title     = {Algorithms for Separable Nonlinear Least Squares Problems},
  journal   = {SIAM Review},
  volume    = {22},
  number    = {3},
  pages     = {318--337},
  year      = {1980},
  publisher = {Society for Industrial and Applied Mathematics},
  url       = {http://www.jstor.org/stable/2030320},
  note      = {Accessed: 2025-06-25}
}

@article{varpro-meta-review,
doi = {10.1088/0266-5611/19/2/201},
url = {https://dx.doi.org/10.1088/0266-5611/19/2/201},
year = {2003},
month = {2},
publisher = {},
volume = {19},
number = {2},
pages = {R1},
author = {Gene Golub and Victor Pereyra},
title = {Separable nonlinear least squares: the variable projection method and its
applications},
journal = {Inverse Problems}
}

@misc{ferreira_local_2010,
	title = {Local convergence analysis of {Gauss}-{Newton}'s method under majorant condition},
	url = {http://arxiv.org/abs/1003.5004},
	doi = {10.48550/arXiv.1003.5004},
	urldate = {2025-12-18},
	publisher = {arXiv},
	author = {Ferreira, O. P. and Goncalves, M. L. N. and Oliveira, P. R.},
	month = mar,
	year = {2010},
	note = {arXiv:1003.5004 [math]},
}

@article{ferreira_kantorovichs_2009,
	title = {Kantorovich’s majorants principle for {Newton}’s method},
	volume = {42},
	copyright = {http://www.springer.com/tdm},
	issn = {0926-6003, 1573-2894},
	url = {http://link.springer.com/10.1007/s10589-007-9082-4},
	doi = {10.1007/s10589-007-9082-4},
	language = {en},
	number = {2},
	urldate = {2025-12-18},
	journal = {Computational Optimization and Applications},
	author = {Ferreira, O. P. and Svaiter, B. F.},
	month = mar,
	year = {2009},
	pages = {213--229},
}

@INPROCEEDINGS{michelena_conv_2025,
  author={Michelena, Santos and Ferreira Da Costa, Maxime  and Picheral, José},
  booktitle={2025 IEEE Statistical Signal Processing Workshop (SSP)}, 
  title={Convergence Guarantees for Unmixing {PSFs} over a Manifold with Non-Convex Optimization}, 
  year={2025},
  volume={},
  number={},
  pages={161-165},
  doi={10.1109/SSP64130.2025.11073360}}

@INPROCEEDINGS{michelena-icassp,
  author={Michelena, Santos and Ferreira Da Costa, Maxime  and Picheral, José},
  booktitle={ICASSP 2026 - 2026 IEEE International Conference on Acoustics, Speech and Signal Processing (ICASSP)}, 
  title={Strong Basin of Attraction for Unmixing Kernels with the Variable Projection Method}, 
  year={2026},
  volume={},
  number={},
  pages={726-730},
  doi={10.1109/ICASSP55912.2026.11461999}}

@article{Smith1992InterlacingSchur,
  author    = {Ronald L. Smith},
  title     = {Some interlacing properties of the Schur complement of a Hermitian matrix},
  journal   = {Linear Algebra and its Applications},
  volume    = {177},
  pages     = {137--144},
  year      = {1992},
  doi       = {10.1016/0024-3795(92)90321-Z},
  url       = {https://doi.org/10.1016/0024-3795(92)90321-Z}
}

@Article{oleary-varpro,
journal={Computational Optimization and Applications},
author={Dianne O'Leary and Bert Rust},
title={Variable projection for nonlinear least squares problems},
year={2013},
month={4},
pages={579-593},
volume={54},
number={3},
doi={10.1007/s10589-012-9492-9},
url={https://ideas.repec.org/a/spr/coopap/v54y2013i3p579-593.html},
}

@article{kaufman1975variable,
  author    = {Linda Kaufman},
  title     = {A variable projection method for solving separable nonlinear least squares problems},
  journal   = {BIT Numerical Mathematics},
  volume    = {15},
  number    = {1},
  pages     = {49--57},
  year      = {1975},
  doi       = {10.1007/BF01932995},
  url       = {https://doi.org/10.1007/BF01932995}
}

@article{poggialini_catching_2023,
	title = {Catching up on calibration-free {LIBS}},
	volume = {8},
	journal = {Journal of Analytical Spectroscopy},
	author = {Poggialini, F. and {colleagues}},
	year = {2023},
	pages = {1234--1245},
}

@book{bauschke2017convex,
  author    = {Heinz H. Bauschke and Patrick L. Combettes},
  title     = {Convex Analysis and Monotone Operator Theory in Hilbert Spaces},
  publisher = {Springer},
  series    = {CMS Books in Mathematics},
  edition   = {2nd},
  year      = {2017},
  isbn      = {9781493911912},
  url       = {https://www.springer.com/gp/book/9781493911912},
}

@misc{salzer2026varpro,
      title={Variable Projection Methods for Solving Regularized Separable Inverse Problems with Applications to Semi-Blind Image Deblurring}, 
      author={Delfina B. Comerso Salzer and Malena I. Español and Gabriela Jeronimo},
      year={2026},
      eprint={2601.05224},
      archivePrefix={arXiv},
      primaryClass={math.NA},
      url={https://arxiv.org/abs/2601.05224}, 
}

@article{Mullen2007,
  author    = {Mullen, K. M. and Vengris, M. and van Stokkum, I. H. M.},
  title     = {Algorithms for separable nonlinear least squares with application to modelling time-resolved spectra},
  journal   = {Journal of Global Optimization},
  volume    = {38},
  number    = {},
  pages     = {201--213},
  year      = {2007},
  doi       = {10.1007/s10898-006-9071-7}
}

@inproceedings{wang2024varproPDE,
  author       = {Yue Wang and Yifan Chen and Xiangxiong Zhang},
  title        = {Variable Projection for Computational PDEs with Artificial Neural Networks},
  booktitle    = {Proceedings of the 2024 Symposium on Scientific Computing and Machine Learning (SCML)},
  year         = {2024},
  url          = {https://scml.jp/2024/paper/10/CameraReady/scml2024.pdf}
}

@ARTICLE{kovacs-var-pro,
  author={Kovács, Péter and Fridli, Sándor and Schipp, Ferenc},
  journal={IEEE Transactions on Signal Processing}, 
  title={Generalized Rational Variable Projection With Application in ECG Compression}, 
  year={2020},
  volume={68},
  number={},
  pages={478-492},
  doi={10.1109/TSP.2019.2961234}}

@article{gauss-newton-kantorovich,
author = {H\"{a}ussler, W. M.},
title = {A Kantorovich-type convergence analysis for the Gauss-Newton-Method},
year = {1986},
issue_date = {January   1986},
publisher = {Springer-Verlag},
address = {Berlin, Heidelberg},
volume = {48},
number = {1},
issn = {0029-599X},
url = {https://doi.org/10.1007/BF01389446},
doi = {10.1007/BF01389446},
journal = {Numer. Math.},
month = jan,
pages = {119–125},
numpages = {7}
}

@book{KantorovichAkilov1982,
  author    = {L. V. Kantorovich and G. P. Akilov},
  title     = {Functional Analysis},
  edition   = {2},
  publisher = {Pergamon Press},
  address   = {Oxford},
  year      = {1982},
  isbn      = {0-08-023036-9},
}

@ARTICLE{chamon,
  author={Chamon, Luiz F. O. and Eldar, Yonina C. and Ribeiro, Alejandro},
  journal={IEEE Transactions on Signal Processing}, 
  title={Functional Nonlinear Sparse Models}, 
  year={2020},
  volume={68},
  number={},
  pages={2449-2463},
  doi={10.1109/TSP.2020.2982834}}

@ARTICLE{nonlin-compress,
  author={Blumensath, Thomas},
  journal={IEEE Transactions on Information Theory}, 
  title={Compressed Sensing With Nonlinear Observations and Related Nonlinear Optimization Problems}, 
  year={2013},
  volume={59},
  number={6},
  pages={3466-3474},
  doi={10.1109/TIT.2013.2245716}}

@article{Beinert_2019,
doi = {10.1088/1361-6420/aaea43},
url = {https://doi.org/10.1088/1361-6420/aaea43},
year = {2018},
month = {11},
publisher = {IOP Publishing},
volume = {35},
number = {1},
pages = {015002},
author = {Beinert, Robert and Bredies, Kristian},
title = {Non-convex regularization of bilinear and quadratic inverse problems by tensorial lifting},
journal = {Inverse Problems}
}

@inproceedings{michelena_gretsi_2025,
  author    = {Santos Michelena and José Picheral and Ferreira Da Costa, Maxime},
  title     = {Reconstruction de Réponses Impulsionnelles sur une Variété pour la Spectroscopie de Plasma Induit par Laser},
  booktitle = {Proceedings of the XXXe Colloque Francophone de Traitement du Signal et des Images (GRETSI 2025)},
  address   = {Strasbourg, France},
  month     = aug,
  year      = {2025},
  organization = {GRETSI}
}

\end{document}